\documentclass{article}
\usepackage{amssymb}
\usepackage{amsmath}
\usepackage{amsthm}
\usepackage{mathtools}
\usepackage{todonotes}
\usepackage{hyperref}
\usepackage{enumerate}
\usepackage{caption}
\usepackage{subcaption,graphicx}
\usepackage{booktabs}

\usepackage{tikz}
\usepackage{pgfplots}
\usepgfplotslibrary{external}
\usepackage[ruled]{algorithm2e}
\usepackage[page]{appendix}

\newtheorem{theorem}{Theorem}
\newtheorem{lemma}[theorem]{Lemma}
\newtheorem{assumption}[theorem]{Assumption}

\theoremstyle{definition}

\newtheorem{notation}[theorem]{Notation}

\newtheorem*{remark}{Remark}

\newtheorem{implementationremark}[theorem]{Implementation Remark}

\makeatletter
\renewcommand{\todo}[2][]{\tikzexternaldisable\@todo[#1]{#2}\tikzexternalenable}
\makeatother

\usepackage{xcolor}

\title{Interpolating Parametrized Quantum Circuits using Blackbox Queries}
\author{\hspace*{-2em}Lars Simon \\ \hspace*{-1em}Bundesdruckerei GmbH \\ {\tt\hspace*{-1em} lars.simon@bdr.de} \and \hspace*{5.em}Holger Eble \\ \hspace*{5.2em}Bundesdruckerei GmbH \\ {\tt \hspace*{4.9em}holger.eble@bdr.de}  \and \hspace*{-4.25em}Hagen-Henrik Kowalski \\ \hspace*{-4.1em}Bundesdruckerei GmbH \\ {\tt \hspace*{-2em}hagen-henrik.kowalski@bdr.de}  \and \hspace*{1.6em}Manuel Radons \\\hspace*{1.5em} Bundesdruckerei GmbH \\ {\tt\hspace*{1.5em} manuel.radons@bdr.de}}
\date{May 2024}

\begin{document}
	
	\maketitle
	
	\begin{abstract}
		This article focuses on developing classical surrogates for parametrized quantum circuits using interpolation via (trigonometric) polynomials. 
We develop two algorithms for the construction of such surrogates and prove performance guarantees. 
The constructions are based on circuit evaluations which are blackbox in the sense that no structural specifics of the circuits are exploited. 
While acknowledging the limitations of the blackbox approach compared to whitebox evaluations, which exploit specific circuit properties, we demonstrate scenarios in which the blackbox approach might prove beneficial.
Sample applications include but are not restricted to the approximation of VQEs and the alleviaton of the barren plateau problem.
	\end{abstract}

	\section{Introduction}\label{sec:introduction}
The goal of this article is to derive classical surrogates of parametrized quantum circuits, \cite{Schreiber_2023}, \cite{jerbi2023shadows}, \cite{landman2022classically}. The notion of {\emph{classical surrogate}} considered in this article is defined more loosely than in the aforementioned references: We want to reliably approximate the expected value of some observable with respect to some state computed by some (parametrized) quantum circuit in dependence of the parameters of the circuit.  As a result, we can consider classical surrogates both in the context of quantum circuit simulation and in the context of quantum machine learning.

In this work we introduce two algorithms for the construction of such surrogates via interpolation by (trigonometric) polynomials. The observable and the parametrized quantum circuit are provided as input to the algorithms, which in turn output classical surrogates in the sense described above. The algorithms require access to an oracle for the evaluation of the expected value at certain points in parameter space. We refer to the oracle queries as blackbox queries/evaluations, since the former exploit no structural specifics of the parametrized quantum circuits. By \emph{samples} we mean data points obtained from blackbox queries, where a data point is a pair consisting of a point in parameter space and expected value at this point. In practice, querying the oracle may involve classical simulation of quantum circuits or execution of quantum circuits on actual quantum hardware. The former can be useful in situations, where there exist certain constellations of points at which the circuit is efficiently classically simulable, while it is not efficiently simulable at arbitrary points in parameter space.

	There exist both advantages and disadvantages of blackbox evaluations.
	The most obvious disadvantage is that whitebox evaluations allow to exploit properties of the individual circuits for improved performance compared to the blackbox approach -- albeit at the cost of tying this superior performance to (potentially very) specific conditions.
	
	If the circuit consists of parametrized Pauli rotations and Clifford gates, approximation of the expected value is possible using methods like Quadratic Clifford Expansion \cite{Mitarai2022}, Sparse Pauli Dynamics \cite{begusic2023fast}, Clifford Perturbation Theory \cite{begusic2023simulating}, Fourier expansion in variational quantum algorithms \cite{Nemkov_2023}, and the LOWESA algorithm \cite{rudolph2023classical}, \cite{fontana2023classical}. Some of these methods have recently been used to successfully simulate IBM's Eagle kicked Ising experiment \cite{Kim2023}.
	
	While all of these approaches are, as we noted above, tied to specific conditions, Clifford gates, together with single-qubit Pauli rotation gates, form a universal gate set. Hence, by the Solovay-Kitaev theorem \cite{Kitaev1997} \cite{KitaevSolovayProof}, any circuit can be efficiently approximated by circuits consisting of such gates. So, by replacing fixed rotation angles by variables, the above-mentioned methods can be applied to simulate such circuits.
	However, there exist straightforward constructions of circuits that make simulation after this fashion infeasible.
		Below we will present a variational circuit similar to the quantum kernel used in  \cite{Havl_ek_2019} whose construction follows such a blueprint, demonstrating that examples which elude the Clifford gate driven approach may, in fact, be practically relevant and are not restricted to some academic fringe. While, of course, a real quantum device is necessary to efficiently obtain samples for such circuits, our algorithms can still be used in this setting.
  
  Our algorithms only involve blackbox queries at certain grid points in the parameter space. Close to these grid points we have strong performance guarantees. This is relevant, since there exist scenarios, in which only the behavior close to an initial point is important. 
We highlight two such cases:

\begin{itemize}
    \item Perhaps the most relevant example for our purposes are VQE, where the ansatz is constructed in such a way that some grid point corresponds to the Hartree-Fock ground state \cite{Mitarai2022}.
    \item Alleviating the barren plateau problem by obtaining a \emph{good enough} approximation of the exact solution via the exact solution of the corresponding problem in the approximate model, which is subsequently used as the initial point for the optimization procedure on a physical quantum device \cite{Mitarai2022}, \cite{Holmes_2022}. 
\end{itemize}
 
For the sake of simplicity we do not take shot noise or the lack of fault-tolerance into account in our theoretical considerations. However, in a follow-up work \cite{simon2024denoising}, we demonstrate experimentally that one of our algorithms can be used to alleviate the adverse effects of both shot noise and (simulated) quantum hardware noise on gradient descent in variational quantum algorithms.

 The worst-case scaling for the accurate simulation of a parametrized circuit through blackbox queries is exponential in the number of parameters, see Lemma 4 in \cite{jerbi2023shadows}. However, in order to ensure that the approximation error vanishes to any (fixed, but arbitrarily high) order around a chosen grid point in parameter space, the number of samples required by our algorithms is polynomial in the number of parameters, see Theorems \ref{thm:performance_guarantees_taylor}, \ref{thm:performance_guarantees_fourier}. In particular, the approximation error is guaranteed to be small in the above-mentioned scenarios, in which only the behavior close to an initial point is important.
	We present two algorithms: Algorithm \ref{algorithm:taylor} computes the Taylor series to an order provided as input using parameter shift rules from quantum machine learning \cite{Mitarai_2018}, \cite{Schuld_2019}, \cite{Mari_2021}, \cite{Wierichs2022generalparameter}. Algorithm \ref{algorithm:fourier} uses a variant of the multivariate Dirichlet kernel to compute an approximation using trigonometric polynomials.\\
	
    \noindent\textbf{Acknowledgement} This article was written as part of the Qu-Gov project, which was commissioned by the German Federal Ministry of Finance. The authors want to extend their gratitude to Manfred Paeschke and Oliver Muth for their continuous encouragement and support.
	
	\subsection{Content and structure}
	
	In Section \ref{sec:algorithms} we go over some preliminaries, introduce Algorithms \ref{algorithm:taylor} and \ref{algorithm:fourier}, state the corresponding performance guarantees in Theorems \ref{thm:performance_guarantees_taylor} and \ref{thm:performance_guarantees_fourier}, and compare the algorithms with existing methods. In Section \ref{sec:experiments} we describe our experiments involving Algorithms \ref{algorithm:taylor} and \ref{algorithm:fourier} and Section \ref{sec:conclusion} contains our closing remarks. Finally, Appendix \ref{appendix:theoretical_background} contains a more rigorous treatment of the theory underlying our algorithms.

	\section{Algorithms}
	\label{sec:algorithms}

In this section we develop the above-mentioned algorithms for classically approximating functions of the form
$$ f\colon \mathbb R^m \to \mathbb R\,,\, \theta\ \mapsto \ \langle\psi (\theta)|\mathcal{M}|\psi (\theta)\rangle\, ,$$
where $|\psi (\theta)\rangle$ is a quantum state for all $\theta$ and $\mathcal{M}$ is an observable. The observable $\mathcal{M}$ is an input provided in Pauli decomposition and the map $\theta\mapsto |\psi (\theta)\rangle$ is an input provided as a parametrized quantum circuit (without measurement). For the technical details of the problem statement, see Section \ref{subsec:tech-prelim}.

\begin{remark}
Theoretically, our algorithms work for arbitrary observables.
However, observables in Pauli decomposition can have up to $4^n$ non-vanishing terms, where $n$ is the number of qubits, which would result in a combinatorial explosion of the computational effort. 
Nevertheless, in practically relevant cases, the number of non-vanishing terms is usually non-prohibitive.
\end{remark}

 \subsection{Technical preliminaries}\label{subsec:tech-prelim}
 Let $n\in \mathbb Z_{>0}$ and $m\in\mathbb{Z}_{\geq 0}$ be integers that denote the number of qubits and parameters, respectively.
For $\theta\in\mathbb{R}^m$, consider the unitary

	\begin{align*}
		U(\theta)
		=C_{m+1}R_m (\theta_m) C_m\cdots R_2 (\theta_2) C_2 R_1 (\theta_1) C_1 
		\in\mathbb{C}^{2^n\times 2^n},
	\end{align*}
	where $C_1 ,\dots , C_{m+1}$ are unitaries given by $n$-qubit quantum circuits and, for all $j\in\{1,\dots ,m\}$, the unitary $R_j (\theta_j)$ is a rotation of the form 
	\begin{align*}
		R_j (\theta_j) = 
		\exp \left(-i\frac{\theta_j}{2}G_j\right)
		\in \mathbb{C}^{2^n\times 2^n}
	\end{align*}
	for some Hermitian $G_j\in \mathbb{C}^{2^n\times 2^n}$ whose set of eigenvalues is $\{-1,1\}$. Moreover, we consider an observable given by a Hermitian matrix $\mathcal{M}\in\mathbb{C}^{2^n\times 2^n}$. Letting $|\psi (\theta)\rangle :=U(\theta)|0\rangle^{\otimes n}$, our aim is to classically approximate the expected value
	\begin{align*}
		f(\theta):=\langle\psi (\theta)|\mathcal{M}|\psi (\theta)\rangle
		,
	\end{align*}
	which defines a function $f\colon\mathbb{R}^m\to\mathbb{R}$. Denoting the Pauli matrices as $I,X,Y,Z$ and recalling that $\{I,X,Y,Z\}^{\otimes n}$ is a basis for the $\mathbb{R}$-vector space of Hermitian matrices in $\mathbb{C}^{2^n\times 2^n}$, we can write
	\begin{align*}
		\mathcal{M}=\sum_{(P_1 ,\dots ,P_n)\in\{I,X,Y,Z\}^n} a_{(P_1 ,\dots ,P_n)} P_1\otimes\cdots\otimes P_n ,
	\end{align*}
	where $a_{(P_1 ,\dots ,P_n)}\in\mathbb{R}$ for all $(P_1 ,\dots ,P_n)\in\{I,X,Y,Z\}^n$. Since this sum consists of $4^n$ summands, we need to assume that the number of non-zero coefficients is small; a common assumption in the field of quantum machine learning is that only $\mathcal{O}(\operatorname{poly}(n))$ of these coefficients are non-zero. In particular, we have
	\begin{align*}
		f(\theta) = 
		\sum_{(P_1 ,\dots ,P_n)\in\{I,X,Y,Z\}^n} a_{(P_1 ,\dots ,P_n)} \langle\psi (\theta)|(P_1\otimes\cdots\otimes P_n )|\psi (\theta)\rangle
	\end{align*}
	for all $\theta\in\mathbb{R}^m$.
	
	\begin{assumption}
		\label{assumption:fast_sampling_of_expected_value_of_pauli_strings}
		We assume that the quantity $\langle\psi (\theta)|(P_1\otimes\cdots\otimes P_n )|\psi (\theta)\rangle$ can be efficiently sampled (either through classical simulation or through evaluation on a quantum device), whenever the components of $\theta$ are integer multiples of $\frac{\pi}{2}$, i.e., whenever $\theta\in\frac{\pi}{2}\mathbb{Z}^m$, and $P_1 ,\dots ,P_n\in\{I,X,Y,Z\}$.
	\end{assumption}
	
	It follows that we can efficiently sample $f$ at points in the grid $\frac{\pi}{2}\mathbb{Z}^m$, assuming that the number of non-zero coefficients in the above decomposition of $\mathcal{M}$ is small. We are particularly interested in the behaviour of our algorithms for when $\theta$ is close to $0\in\mathbb{R}^m$.
	
	Before describing our algorithms, we briefly point out that we can write
	\begin{align*}
		f(\theta) =\sum_{\omega\in\{-1,0,1\}^m} c_\omega e^{i\omega^T \theta}
		\text{ for all }\theta\in\mathbb{R}^m,
	\end{align*}
	where $c_\omega\in\mathbb{C}$ and $c_\omega = \overline{c_{-\omega}}$ for all $\omega\in\{-1,0,1\}^m$, see \cite{Schuld_2021}. Here, $\omega^T\in\mathbb{R}^{1\times m}$ denotes the matrix transpose of $\omega\in\mathbb{R}^m\cong\mathbb{R}^{m\times 1}$, i.e., $\omega^T \theta = \omega_1\theta_1 +\dots + \omega_m\theta_m$.
	
	\begin{remark}
		We imposed quite restrictive assumptions on the set of eigenvalues of $G_1 ,\dots , G_m$, as well as on the set of points in $\mathbb{R}^m$ where $f$ can be sampled efficiently. This was done for ease of presentation -- in fact, adjusting our algorithms to more general eigenvalue spectra and more general subsets of $\mathbb{R}^m$ on which $f$ can be sampled efficiently is straightforward.
	\end{remark}
	
	\subsection{Algorithm based on Taylor polynomials}
	
	In \cite{Mitarai2022}, the first- and second-order partial derivatives were computed by passing certain Pauli operators through Clifford circuits and exploiting the classical simulatability of such circuits  \cite{Gottesman1998} \cite{VanDenNes2010}. In this section we compute arbitrary-order partial derivatives in the more general setting described above. We will do so using explicit formulas for the partial derivatives which are known as parameter shift rules in the context of quantum machine learning, see \cite{Mitarai_2018}, \cite{Schuld_2019}, \cite{Mari_2021}, \cite{Wierichs2022generalparameter}.
	
	\begin{lemma}
		\label{lemma:partials_for_our_function_f}
		For all multiindices $\alpha\in (\mathbb{Z}_{\geq 0})^m$ we have:
		\begin{align*}
			D^{\alpha}f (0) = \frac{1}{2^{|\alpha |}} \sum_{\mathfrak{i}\in\{-1,1\}^{|\alpha |}} \mathfrak{i}^{(1,\dots ,1)} f(p_{\alpha ,\mathfrak{i}}) ,
		\end{align*}
		where we index the entries of $\mathfrak{i}\in\{-1,1\}^{|\alpha |}$ as
		\begin{align*}
			\mathfrak{i} = \left(\mathfrak{i}_{1,1},\dots ,\mathfrak{i}_{1,\alpha_1},\dots ,\mathfrak{i}_{m,1},\dots ,\mathfrak{i}_{m,\alpha_m}\right),
		\end{align*}
		and
		\begin{align*}
			p_{\alpha ,\mathfrak{i}} = \frac{\pi}{2}\cdot\left((\mathfrak{i}_{1,1}+\dots +\mathfrak{i}_{1,\alpha_1})\bmod 4,\dots ,(\mathfrak{i}_{m,1}+\dots +\mathfrak{i}_{m,\alpha_m})\bmod 4\right).
		\end{align*}
		 In particular, all terms appearing in the above sum can be evaluated efficiently, assuming that the number of non-zero coefficients in the decomposition of $\mathcal{M}$ is small.
	\end{lemma}
	\begin{proof}
		For the proof of the formula for $D^{\alpha}f(0)$ we refer to Notation \ref{notation:shiftvectors} and Lemma \ref{lemma:parametershift} in Appendix \ref{appendix:theoretical_background}. Regarding efficient evaluation, we observe that $p_{\alpha, \mathfrak{i}}\in \frac{\pi}{2}\mathbb{Z}^m$, so the claim follows from Assumption \ref{assumption:fast_sampling_of_expected_value_of_pauli_strings}.
	\end{proof}

	So, in order to approximate $f$ by its Taylor polynomial of order $L\in\mathbb{Z}_{\geq 0}$, we compute all partial derivatives $D^{\alpha}f (0)$ with $|\alpha |\leq L$ using Lemma \ref{lemma:partials_for_our_function_f} and subsequently use the approximation
	\begin{align*}
		f(\theta_1 ,\dots ,\theta_m )\approx \sum_{\alpha\in (\mathbb{Z}_{\geq 0})^m\text{ with }|\alpha |\leq L}\frac{D^\alpha f (0)}{\alpha_1 !\cdots \alpha_m !} \theta_1^{\alpha_1}\cdots \theta_m^{\alpha_m}.
	\end{align*}
	
	\begin{implementationremark}
		\label{implementationremark:store_values}
		Note that, with the notation from Lemma \ref{lemma:partials_for_our_function_f}, it will often be the case that $p_{\alpha ,\mathfrak{i}}=p_{\beta ,\mathfrak{j}}$, but $(\alpha ,\mathfrak{i})\neq (\beta ,\mathfrak{j})$. So, by storing the values for the points at which $f$ has already been evaluated, we reduce the number of times we have to sample $f$.
	\end{implementationremark}

	The pseudocode for this algorithm can be found in Algorithm \ref{algorithm:taylor}. In Theorem \ref{thm:performance_guarantees_taylor} we give some performance guarantees for Algorithm \ref{algorithm:taylor}. The proof can be found in Appendix \ref{appendix:proof_taylor}.
	
		\begin{algorithm}
		\SetKwInOut{Input}{Input}
		\SetKwInOut{Output}{Output}
		
		\Input{order $L$, observable $\mathcal{M}$, parametrized unitary $U$}
		\Output{approximation $\tilde{f}$ for $f$, where $f(\theta)= \langle0^n|U^\dag (\theta)\mathcal{M}U(\theta)|0^n\rangle$}
		
		Decompose $\mathcal{M}=\sum_{t=1}^{T} a_t \mathcal{Q}_t$, where $a_t\in\mathbb{R}$ coefficient and $\mathcal{Q}_t\in\{I,X,Y,Z\}^{\otimes n}$ Pauli operator for all $t$
		
		\For{$\alpha\in (\mathbb{Z}_{\geq 0})^m$ with $|\alpha |\leq L$}
			{
				Compute
				\begin{align*}
					D^{\alpha}f (0) = \frac{1}{2^{|\alpha |}}  \sum_{\mathfrak{i}\in\{-1,1\}^{|\alpha |}} \mathfrak{i}^{(1,\dots ,1)} \sum_{t=1}^{T} a_t \langle0^n|U^\dag (p_{\alpha ,\mathfrak{i}})\mathcal{Q}_tU(p_{\alpha ,\mathfrak{i}})|0^n\rangle
				\end{align*}
				using Assumption \ref{assumption:fast_sampling_of_expected_value_of_pauli_strings} and Implementation Remark \ref{implementationremark:store_values}
			}

		Set
		\begin{align*}
			\tilde{f}(\theta):= \sum_{\alpha\in (\mathbb{Z}_{\geq 0})^m\text{ with }|\alpha |\leq L}\frac{D^\alpha f (0)}{\alpha_1 !\cdots \alpha_m !} \theta_1^{\alpha_1}\cdots \theta_m^{\alpha_m}
		\end{align*}
		
		\Return $\tilde{f}$
		
		\caption{Approximation with Taylor polynomial}
		\label{algorithm:taylor}
	\end{algorithm}
	
	\begin{theorem}
		\label{thm:performance_guarantees_taylor}
		Assume Algorithm \ref{algorithm:taylor} is executed with $L\in\mathbb{Z}_{\geq 0}$, $\mathcal{M}$ and $U$, where the latter two are as in the beginning of Section \ref{sec:algorithms}. Let $f\colon\mathbb{R}^m\to\mathbb{R}$ be as above and let  $\tilde{f}\colon\mathbb{R}^m\to\mathbb{R}$ be the output of the algorithm. Then the following holds:
		\begin{enumerate}[(i)]
			\item \label{thm:performance_guarantees_taylor_number_f_eval} If $L\leq m$, then, during execution of the algorithm, $f$ is sampled at no more than $\frac{4^L}{L!}m^L$ points in $\frac{\pi}{2}\mathbb{Z}^m$,
			\item \label{thm:performance_guarantees_taylor_error_estimate} Writing $\mathcal{M}=\sum_{(P_1 ,\dots ,P_n)\in\{I,X,Y,Z\}^n} a_{(P_1 ,\dots ,P_n)} P_1\otimes\cdots\otimes P_n$ as in the beginning of Section \ref{sec:algorithms}, we have the following bound on the approximation error for all $\theta\in\mathbb{R}^m$:
			\begin{align*}
				\left|\tilde{f}(\theta) -f(\theta)\right|\leq & 
				\left(
				\sum_{(P_1 ,\dots ,P_n)\in\{I,X,Y,Z\}^n} \left| a_{(P_1 ,\dots ,P_n)}\right|
				\right)\\
				& \cdot
				\left(\exp\left(\Vert \theta\Vert_1\right)-\sum_{k=0}^{L} \frac{\Vert \theta\Vert_1^k}{k!}\right) 
				,
			\end{align*}
			where $\Vert\cdot\Vert_1$ denotes the $1$-norm on $\mathbb{R}^m$. In particular, if $\Vert \theta\Vert_1\leq 1 +\frac{L}{2}$, we have:
			\begin{align*}
				\left|\tilde{f}(\theta) -f(\theta)\right|
				\leq
				2
				\left(
				\sum_{(P_1 ,\dots ,P_n)\in\{I,X,Y,Z\}^n} \left| a_{(P_1 ,\dots ,P_n)}\right|
				\right)
				\frac{\Vert \theta\Vert_1^{L+1}}{(L+1)!}
				.
			\end{align*}
		\end{enumerate}
	\end{theorem}
	
	\subsection{Algorithm based on  trigonometric polynomials}

	Recall that we can write
	\begin{align*}
		f(\theta) =\sum_{\omega\in\{-1,0,1\}^m} c_\omega e^{i\omega^T \theta}
		\text{ for all }\theta\in\mathbb{R}^m,
	\end{align*}
	where $c_\omega\in\mathbb{C}$ and $c_\omega = \overline{c_{-\omega}}$ for all $\omega\in\{-1,0,1\}^m$. We denote the set of all functions of this form as $H$ and give $H$ the structure of a reproducing kernel Hilbert space with reproducing kernel $K$ by defining
	\begin{align*}
		\langle g_1 ,g_2\rangle_H= \int_{[-\pi ,\pi]^m} g_1 (z) g_2 (z) \mathrm{d}z, \ g_1,g_2\in H,
	\end{align*}
	and
	\begin{align}\label{eq:kernel_definition}
		K(x,z) = \frac{1}{(2\pi )^m} \prod_{j=1}^m (1+2\cos(x_j - z_j )), \ x,z\in\mathbb{R}^m.
	\end{align}
	For a more rigorous treatment of these facts, see Appendix \ref{appendix:theoretical_background}. The idea is now to sample $f$ at certain points in $\frac{\pi}{2}\mathbb{Z}^m$ using Assumption \ref{assumption:fast_sampling_of_expected_value_of_pauli_strings} and to subsequently interpolate with elements of $H$ by exploiting the reproducing kernel Hilbert space structure.
	
	So let $L\in\mathbb{Z}$, $0\leq L\leq m$, be the desired order of the approximation, which will be provided as an input to the algorithm. We then determine all points $p_1 ,\dots , p_D\in \frac{\pi}{2}\{-1,0,1\}^m$ with the property that at most $L$ entries are non-zero. More formally, $\{p_1 ,\dots , p_D\} = \left(\frac{\pi}{2}\{-1, 0, 1\}^m\right)\cap\mathcal{I}_L$, where $\mathcal{I}_L = \left\{z\in\mathbb{R}^m\vert \operatorname{card}(\{j\in\{1,\dots ,m\}\vert z_j\neq 0\})\leq L\right\}$, and $p_i\neq p_j$ whenever $i\neq j$. We then determine $f(p_1),\dots ,f(p_D)$ using Assumption \ref{assumption:fast_sampling_of_expected_value_of_pauli_strings} and subsequently obtain coefficients $\eta_1 ,\dots ,\eta_D\in\mathbb{R}$ by solving the linear system of equations
	\begin{align}\label{eq:kernel_lgs}
		\left(K(p_i ,p_j )\right)_{1\leq i,j\leq D}\cdot\eta =
		\begin{pmatrix}
			f(p_1) \\
			\vdots\\
			f(p_D)
		\end{pmatrix}
		, \text{ where } \eta\in\mathbb{R}^D .
	\end{align}
	
	\begin{lemma}
		\label{lemma:lgs_has_unique_solution}
		The system (\ref{eq:kernel_lgs}) of linear equations has a uniquely determined solution $\eta\in\mathbb{R}^D$.
	\end{lemma}
	\begin{proof}
		The elements $p_1 ,\dots ,p_D$ of $\frac{\pi}{2}\{-1,0,1\}^m$ are pairwise distinct. So, by Lemma \ref{lemma:basis_for_H} in Appendix \ref{appendix:theoretical_background}, the matrix $\left(K(p_i ,p_j )\right)_{1\leq i,j\leq D}\in\mathbb{R}^{D\times D}$ is the Gram matrix of $D$ linearly independent vectors in $H$ and hence invertible. The claim follows.
	\end{proof}
	
	We then use the approximation
	\begin{align*}
		f(\theta_1 ,\dots ,\theta_m )\approx\sum_{j=1}^{D} \eta_j K(p_j, (\theta_1 ,\dots ,\theta_m )) .
	\end{align*}
	Note that the approximation coincides with $f$ on $\{p_1 ,\dots ,p_D\}$, since $\eta$ solves the above linear system of equations.
	
	\begin{implementationremark}
		\label{implementationremark:K_not_numerically_stable}
		If $m$ is large, then evaluating $K$ according to formula (\ref{eq:kernel_definition}) will not be numerically stable, since $\frac{1}{(2\pi )^m}$ will be very close to $0$. For our numerical experiments we instead implemented
		\begin{align*}
			\tilde{K}(x,z) = \prod_{j=1}^m \frac{1+2\cos(x_j - z_j )}{3}.
		\end{align*}
		Note that $\tilde{K} = \left(\frac{2\pi}{3}\right)^m K$, so the error we are making is a constant factor, which will automatically be absorbed into the coefficient vector $\eta$ when solving the linear system of equations. In fact, $\tilde{\eta} = \left(\frac{3}{2\pi}\right)^m \eta$ and thus $\sum_{j=1}^{D}\tilde{\eta}_j \tilde{K}(p_j,\cdot ) = \sum_{j=1}^{D}\eta_j K(p_j,\cdot )$. 
	\end{implementationremark}

	The pseudocode for this algorithm can be found in Algorithm \ref{algorithm:fourier}. In Theorem \ref{thm:performance_guarantees_fourier} we give some performance guarantees for Algorithm \ref{algorithm:fourier}. The proof can be found in Appendix \ref{appendix:proof_fourier}.
	
	\begin{algorithm}
		\SetKwInOut{Input}{Input}
		\SetKwInOut{Output}{Output}
		
		\Input{order $L$, observable $\mathcal{M}$, parametrized unitary $U$}
		\Output{approximation $\tilde{f}$ for $f$, where $f(\theta)= \langle0^n|U^\dag (\theta)\mathcal{M}U(\theta)|0^n\rangle$}
		
		Decompose $\mathcal{M}=\sum_{t=1}^{T} a_t \mathcal{Q}_t$, where $a_t\in\mathbb{R}$ coefficient and $\mathcal{Q}_t\in\{I,X,Y,Z\}^{\otimes n}$ Pauli operator for all $t$
		
		Determine the set $\left(\frac{\pi}{2}\{-1, 0, 1\}^m\right)\cap\mathcal{I}_L$, see Lemma \ref{lemma:fewer_tuples_enough_for_fourier} in Appendix \ref{appendix:theoretical_background}, and denote its elements as $p_1 ,\dots , p_D$ with $p_i\neq p_j$ for $i\neq j$
		
		\For{$j=1,\dots ,D$}
		{
			Compute
			\begin{align*}
				f(p_j)= \sum_{t=1}^{T} a_t \langle0^n|U^\dag (p_j)\mathcal{Q}_t U(p_j)|0^n\rangle
			\end{align*}
			using Assumption \ref{assumption:fast_sampling_of_expected_value_of_pauli_strings}
		}
		
		Find the uniquely determined $\eta\in\mathbb{R}^D$ (see Lemma \ref{lemma:lgs_has_unique_solution} and Implementation Remark \ref{implementationremark:K_not_numerically_stable}) solving the linear system of equations
		\begin{align*}
			\left(K(p_i ,p_j )\right)_{1\leq i,j\leq D}\cdot\eta =
			\begin{pmatrix}
				f(p_1) \\
				\vdots\\
				f(p_D)
			\end{pmatrix}
		\end{align*}
		
		Set
		\begin{align*}
			\tilde{f}(\theta):= \sum_{j=1}^{D} \eta_j K(p_j, \theta)
		\end{align*}
		
		\Return $\tilde{f}$
		
		\caption{Approximation with trigonometric polynomial}
		\label{algorithm:fourier}
	\end{algorithm}
	
	\begin{theorem}
		\label{thm:performance_guarantees_fourier}
		Assume Algorithm \ref{algorithm:fourier} is executed with $L\in\{0,1,\dots ,m\}$, $\mathcal{M}$ and $U$, where $m$, $\mathcal{M}$, and $U$ are as in the beginning of Section \ref{sec:algorithms}. Let $f\colon\mathbb{R}^m\to\mathbb{R}$ be as above and let  $\tilde{f}\colon\mathbb{R}^m\to\mathbb{R}$ be the output of the algorithm. Then the following holds:
		\begin{enumerate}[(i)]
			\item \label{thm:performance_guarantees_fourier_number_f_eval} During execution of the algorithm, $f$ is sampled at no more than $\frac{3^L}{L!}m^L$ points in $\frac{\pi}{2}\mathbb{Z}^m$.
			\item \label{thm:performance_guarantees_fourier_projection} $\tilde{f}\in H$. In fact, $\tilde{f}$ is the orthogonal projection of $f$ onto
			\begin{align*}
				\operatorname{span}_{\mathbb{R}}\left(\left\{K_q\Big\vert q\in \left(\frac{\pi}{2}\{-1, 0, 1\}^m\right)\cap\mathcal{I}_L\right\}\right)
			\end{align*}
			(see Lemma \ref{lemma:fewer_tuples_enough_for_fourier} in Appendix \ref{appendix:theoretical_background}) wrt.\ $\langle\cdot ,\cdot\rangle_H$.
			\item \label{thm:performance_guarantees_fourier_minimal_norm} $\tilde{f}$ is the uniquely determined minimal-norm element of $H$ which agrees with $f$ on $\left(\frac{\pi}{2}\{-1, 0, 1\}^m\right)\cap\mathcal{I}_L$.
			\item \label{thm:performance_guarantees_fourier_coincides_on_I_L} $\tilde{f}$ coincides with $f$ on $\mathcal{I}_L$. In particular, $\tilde{f}$ coincides with $f$ on any subspace of $\mathbb{R}^m$ that is spanned by at most $L$ of the coordinate axes.
			\item \label{thm:performance_guarantees_fourier_coordinate_axes} If $L\geq 1$, then $\tilde{f}$ coincides with $f$ on the coordinate axes in $\mathbb{R}^m$.
			\item \label{thm:performance_guarantees_fourier_whole_space} If $L=m$, then $\tilde{f}$ coincides with $f$ on all of $\mathbb{R}^m$, i.e., we completely recover the function $f$.
			\item \label{thm:performance_guarantees_fourier_error_estimate}
			For all $\alpha\in (\mathbb{Z}_{\geq 0})^m$ with $|\alpha |\leq L$ we have
			\begin{align*}
				D^{\alpha}(f-\tilde{f})(0) = 0.
			\end{align*}
			In particular, the power series expansions of $f$ and $\tilde{f}$ around $0\in\mathbb{R}^m$ coincide up to order at least $L$. 
		\end{enumerate}
	\end{theorem}

	\begin{remark}
		In the setting of Theorem \ref{thm:performance_guarantees_fourier}, it is easy to obtain crude explicit bounds on the approximation error $|\tilde{f}(\theta)-f(\theta)|$ for $\theta\in\mathbb{R}^m$ by coarsely estimating the coefficients in the power series expansion of $\tilde{f}-f$ using (\ref{thm:performance_guarantees_fourier_error_estimate}) in Theorem \ref{thm:performance_guarantees_fourier} and Lemma \ref{lemma:parametershift} in Appendix \ref{appendix:theoretical_background}. Since we believe these bounds to be of limited use in practice, we do not state them here. However, we are confident that stronger bounds can be established using some of the ideas from Appendix \ref{appendix:theoretical_background}. We leave this for future work.
	\end{remark}
	
	\subsection{Comparison with existing methods}

	We now compare our algorithms to some of the existing methods mentioned in Section \ref{sec:introduction}.

	\subsubsection{Quantum circuit simulation}
	Our algorithms, as well as the methods mentioned in Section \ref{sec:introduction}, deal with the approximation (resp. computation) of the expected value of some observable, which -- strictly speaking -- is less general than quantum circuit simulation.
	Our algorithms require access to an oracle, see Assumption \ref{assumption:fast_sampling_of_expected_value_of_pauli_strings}. In contrast, the method introduced in \cite{Mitarai2022} and the various methods \cite{begusic2023fast},  \cite{begusic2023simulating}, \cite{Nemkov_2023}, \cite{rudolph2023classical}, \cite{fontana2023classical} used to simulate IBM's Eagle kicked Ising experiment \cite{Kim2023} do \emph{not} require access to such an oracle. Our algorithms could also be used to simulate this experiment, since, for $\theta\in\frac{\pi}{2}\mathbb{Z}^m$, all gates in the occurring circuits are Clifford (i.e., Assumption \ref{assumption:fast_sampling_of_expected_value_of_pauli_strings} is satisfied by the Gottesman-Knill theorem \cite{Gottesman1998} \cite{VanDenNes2010}). However, the previously mentioned methods are \textit{whitebox} methods, i.e., they exploit the internal structure of the circuit, whereas our algorithms only make use of the values of $f$ at certain grid points and assume no knowledge about the internal structure of the circuit. Consequently, one would expect our algorithms to perform significantly worse than most of the previously mentioned methods when it comes to the simulation of circuits consisting only of parametrized Pauli rotations and Clifford gates. However, these methods are less general than our algorithms, in the sense that we do not assume the non-parametrized gates to be Clifford.

	\subsubsection{Classical surrogates of quantum machine learning models}

	Algorithm \ref{algorithm:fourier} carries some similarities to the methods presented in \cite{Schreiber_2023} and \cite{landman2022classically}, which both also make use of the fact that the function computed by a variational quantum circuit can (under suitable assumptions on the parametrized gates) be written as a partial Fourier series \cite{Schuld_2021}. While we impose more restrictive assumptions on the parametrized quantum circuits (specifically, no data re-uploading and less general eigenvalue spectra for the generators of the parametrized gates), our methods can readily be adapted to work under the less restrictive assumptions considered in \cite{Schreiber_2023} and \cite{landman2022classically} respectively.
	
	In the setting considered in our work, the method introduced in \cite{Schreiber_2023} needs to query the trained quantum model at $3^m$ grid points. Ignoring the effect of measurement shot noise and quantum hardware noise, this method and Algorithm \ref{algorithm:fourier} are both guaranteed to perfectly recover the quantum model using $3^m$ grid points. However, as $m$ gets large, querying the quantum model at $3^m$ grid points quickly becomes prohibitive. In contrast to the method introduced in \cite{Schreiber_2023}, Algorithm \ref{algorithm:fourier} can still be applied in this scenario (by choosing an order $L$ less than $m$). While perfect recovery of the quantum model is no longer possible in this case, one is still guaranteed to obtain a good approximation close to a grid point, see Theorem \ref{thm:performance_guarantees_fourier}.
	
	Finally, we compare Algorithm \ref{algorithm:fourier} to the the method introduced in \cite{landman2022classically}. Instead of using a trained quantum model as an oracle, this method involves training a classical model -- derived from the architecture of the quantum model using Random Fourier Features \cite{random_fourier_features} -- on the original training data. As such, this method is qualitatively different from Algorithm \ref{algorithm:fourier}. However, there are some similarities; most importantly, the classical approximation is obtained using kernel ridge regression. In \cite{landman2022classically}, the occurring kernel is approximated using Random Fourier Features, whereas Algorithm \ref{algorithm:fourier} uses the kernel $K$ directly (note that $K$ can be efficiently evaluated classically).

	\section{Experiments}
	\label{sec:experiments}

    In this section, we describe our experiments with Algorithms \ref{algorithm:taylor} and \ref{algorithm:fourier}. For our experiments we chose a small, albeit representative parametrized quantum circuit obtained from the following construction:

    Consider an $n$-qubit quantum circuit with initial state $|0\rangle^{\otimes n}$. 
    The subsequent construction pairwise entangles all qubits and introduces a number of $T$-gates which is quadratic in the number of qubits. Repeat the following $d\geq 2$ times: First, apply parametrized $RX$-gates on each of the $n$ qubits. Then, for qubit pairs $(i, j)$, where $i=1,\dots , n-1$ and $j=i+1,\dots ,n$ and these pairs are listed in the lexicographic order, we apply a $CNOT$-gate with control $i$ and target $j$, a $T$-gate on qubit $j$, and another $CNOT$-gate with control $i$ and target $j$. Finally, we determine the expected value of the observable $Z^{\otimes n}$. The dimension of the parameter space $\mathbb{R}^m$ is then $m=nd$ and the number of $T$-gates is $\frac{dn(n-1)}{2}$. For a visual representation of the circuit, see Figure \ref{fig:circuit}.

\begin{figure}[htp]
	\centering
	\subcaptionbox{circuit}{\includegraphics[width=4.7in]{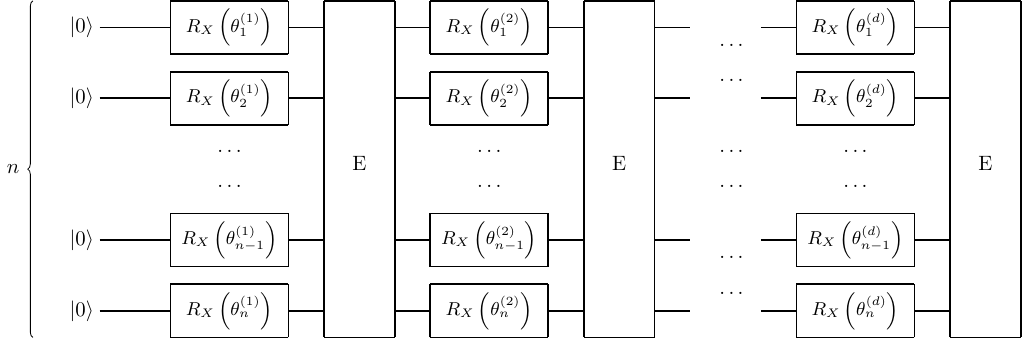}}
	
	\bigskip
	
	\centering
	\subcaptionbox{entangling block}{\includegraphics[width=4.7in]{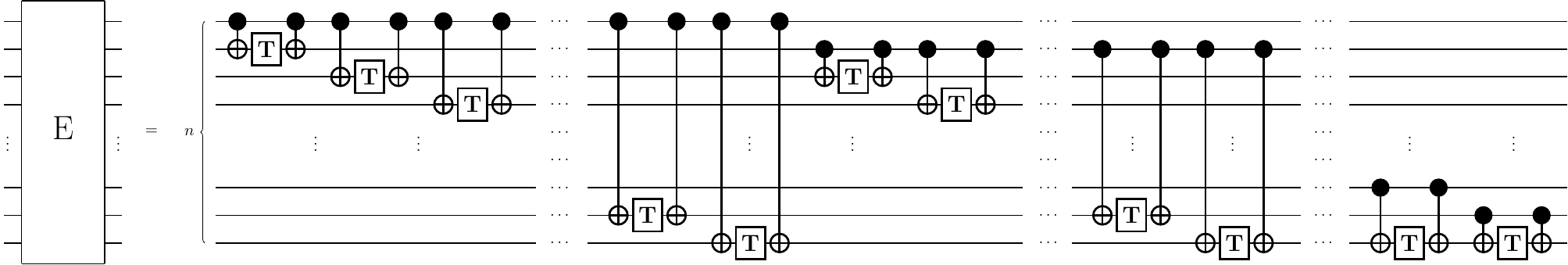}}

	\caption{This figure shows the circuit used in the experiments in Section \ref{sec:experiments}.}
	\label{fig:circuit}
\end{figure}

    This circuit is similar to circuits which are relevant in practice, since it is inspired by the quantum kernel featuring in \cite{Havl_ek_2019}. However, attempting to simulate this circuit by first replacing the $T$ gates by parametrized $RZ$ gates and subsequently using one of the whitebox methods for simulating circuits consisting of Pauli rotations and Clifford gates mentioned in the introduction, is infeasible because of the unfavourable scaling of the number of $T$ gates wrt.\ $n$. However, since $m$ scales more favourably with $n$, Algorithms \ref{algorithm:taylor} and \ref{algorithm:fourier} can be applied to construct a classical surrogate, assuming that samples can be obtained from a real quantum device.

    For our experiments we chose $n=8$ and $d=2$, which yields a paramteter space dimension $m=n d=16$. While the circuit is small and as such can be simulated exactly using standard methods, the local behaviour around grid points we observe in the experiments is still representative (even for large values of $n$ and $d$ where classical simulation is no longer feasible), since this behaviour is guaranteed by Theorems \ref{thm:performance_guarantees_taylor} and \ref{thm:performance_guarantees_fourier}. In Section \ref{subsec:behavioralongcurves} we analyse the behavior of the surrogates computed by Algorithms \ref{algorithm:taylor} and \ref{algorithm:fourier} along various  curves in parameter space. In Section \ref{subsec:montecarlonorm} we analyse the approximation error of the surrgoates in $L^2$-norm.

\subsection{Behavior along curves in parameter space}\label{subsec:behavioralongcurves}
Figures \ref{fig:default} and \ref{fig:special} illustrate the behavior of the surrogates computed by Algorithms \ref{algorithm:taylor} and \ref{algorithm:fourier}. 
Figure \ref{fig:default} focuses on the local behavior in the neighborhood of a grid point.
Figure \ref{fig:special} focuses on special aspects of the behaviors of the computed approximations.
Since the parameter space is $16$-dimensional, we plot the behavior along a set of curves in the latter. These are:

\begin{enumerate}
    \item The diagonal in $\mathbb{R}^m$ through grid point $0$; curve $\gamma_1\colon \mathbb{R}\to\mathbb{R}^{16}$, $t\mapsto (t,\dots ,t)$.
    \item A curve which is tangent to the first coordinate axis in $\mathbb{R}^m$; curve $\gamma_2\colon \mathbb{R}\to\mathbb{R}^{16}$, $t\mapsto\frac{\pi}{2}(\sin (t), (1-\cos (t))^2, \sin^4(t),\dots ,\sin^4(t))$. Here we expect good approximation behaviour because of Theorems \ref{thm:performance_guarantees_taylor} and \ref{thm:performance_guarantees_fourier}. In general, the higher the order of tangency to certain subspaces of $\mathbb{R}^m$ spanned by coordinate axes, the better the expected approximation behavior.
    \item A curve with constant $1$-norm in $\mathbb{R}^m$. We look at this because the error estimates in Theorem \ref{thm:performance_guarantees_taylor} are stated in terms of the $1$-norm. Curve $\gamma_3\colon [-1,1]\to\mathbb{R}^{16}$, $t\mapsto\frac{4\pi}{5} \left(\frac{t+1}{30},\dots , \frac{t+1}{30}, 1-\frac{t+1}{2}\right)$.
    \item The curve $\gamma_4\colon [-1,1]\to\mathbb{R}^{16}$, $t\mapsto\frac{2\pi}{5} \left(\frac{t+1}{30},\dots , \frac{t+1}{30}, 1-\frac{t+1}{2}\right)$.
    \item A curve that lies in a $4$-dimensional subspace of parameter space spanned by four of the coordinate axes, $\gamma_5\colon\mathbb{R}\to\mathbb{R}^{16}$, $t\mapsto (t,t,t,t,0,\dots ,0)$. The fifth curve is used exclusively in the third row of Figure \ref{fig:special} to hint at the global behavior of the surrogate computed by Algorithm \ref{algorithm:fourier}.
\end{enumerate}

In Figure \ref{fig:default} we plot the approximations computed by Algorithms \ref{algorithm:taylor} and \ref{algorithm:fourier} along the first three of these curves for parameter values $L=1,2,4$ and $L=1,2,3$, respectively. For Algorithm \ref{algorithm:taylor} we omit $L=3$ since there are no Taylor terms of order $3$.
The rows $1$ to $3$ of the figure show the behavior along curves $\gamma_1$ to $\gamma_3$.

In Figure \ref{fig:special} we focus on the specific properties of both algorithms. 
The subfigures on the left (Algorithm \ref{algorithm:taylor}) contain only one plot for $L=4$, displaying the local behavior along curves $\gamma_1$, $\gamma_2$ and $\gamma_4$ (in rows $1,2,3$, respectively). The plots are underlayed with the intervals given by the performance guarantees, see Theorem \ref{thm:performance_guarantees_taylor}. 

For Algorithm \ref{algorithm:fourier} we try to give an idea of the surrogate's global behavior. As such, the set $\{p_1 ,\dots ,p_D\}$ featuring in Algorithm \ref{algorithm:fourier} was enriched by the corresponding points for the grid point $(\pi /2 ,\dots ,\pi /2 )$.
In rows $1$ and $2$ we plot the behavior along the curves $\gamma_1$ and $\gamma_2$, respectively. Unlike in Figure \ref{fig:default}, the values corresponding to a second grid point $(\pi /2 ,\dots ,\pi /2)$ are shown in the figure. In the third row we plot the behavior in a $4$-dimensional subspace spanned by coordinate axes, using $\gamma_5$. While we restrict ourselves to a low-dimensional subspace of parameter space, the plot still hints at the global behavior of the surrogate insofar as the curve moves quite far away from the initial grid point $(0,\dots,0)$. The good approximation quality is expected in light of \ref{thm:performance_guarantees_fourier_coincides_on_I_L} in Theorem \ref{thm:performance_guarantees_fourier}.

The plots featuring in Figures \ref{fig:default} and \ref{fig:special} clearly mirror the performance guarantees from Theorems \ref{thm:performance_guarantees_taylor} and \ref{thm:performance_guarantees_fourier}.

\begin{figure}[htp]
	\centering
\subcaptionbox{Algorithm \ref{algorithm:taylor}, plot along curve $\gamma_1$\label{fig:default_taylor_diagonal}}{\includegraphics[width=2.3in]{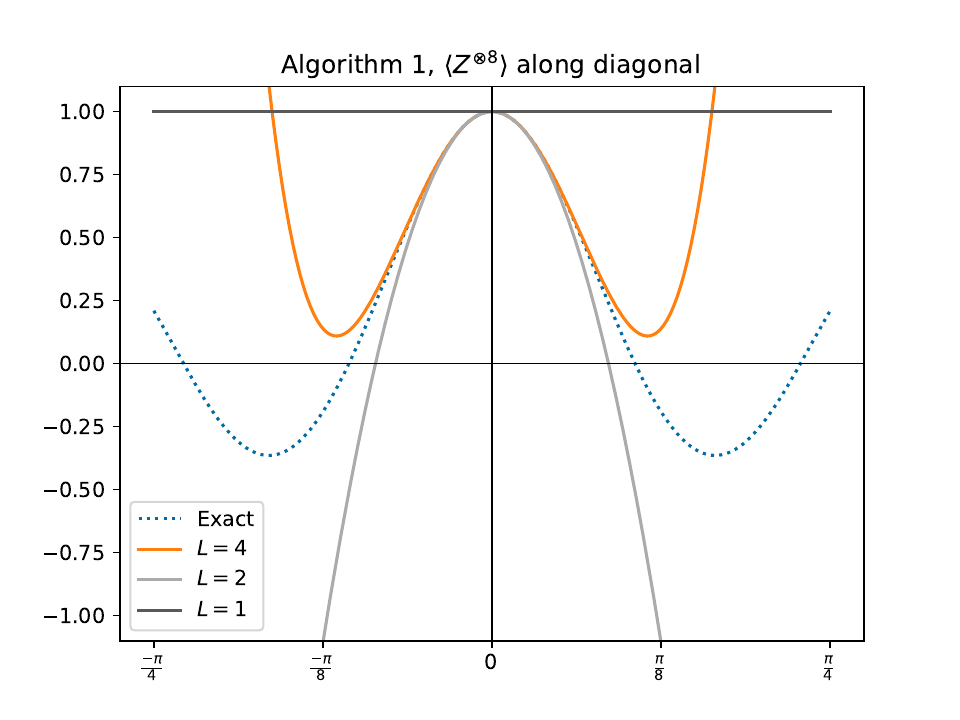}}\hspace{1em}%
\subcaptionbox{Algorithm \ref{algorithm:fourier}, plot along curve $\gamma_1$\label{fig:default_fourier_diagonal}}{\includegraphics[width=2.3in]{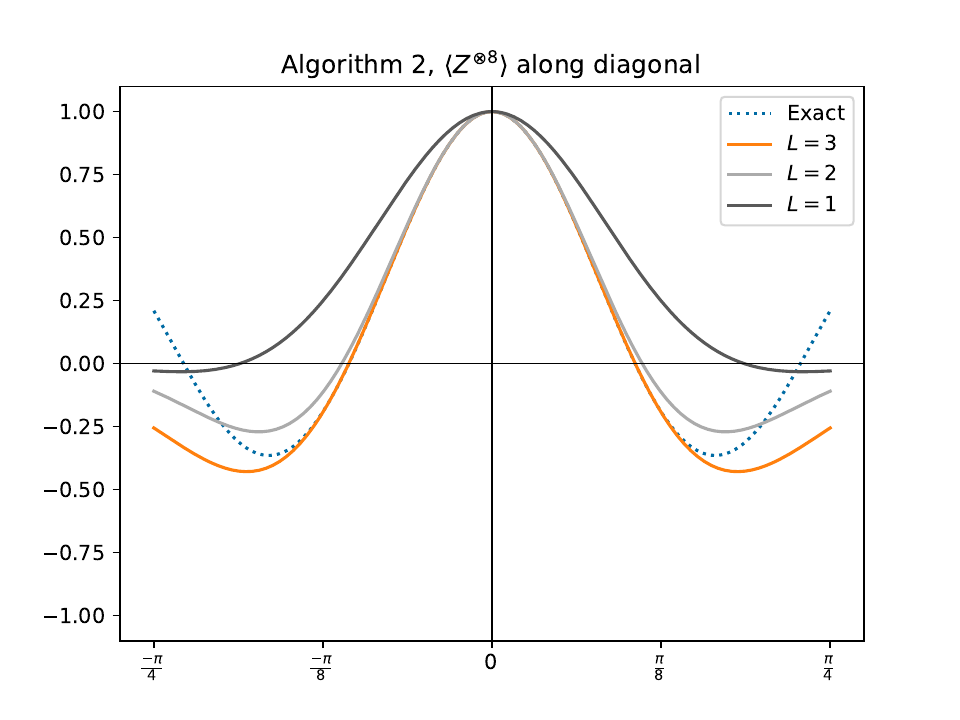}}
	
	\bigskip
	
		\centering
	\subcaptionbox{Algorithm \ref{algorithm:taylor}, plot along curve $\gamma_2$\label{fig:default_taylor_tangent}}{\includegraphics[width=2.3in]{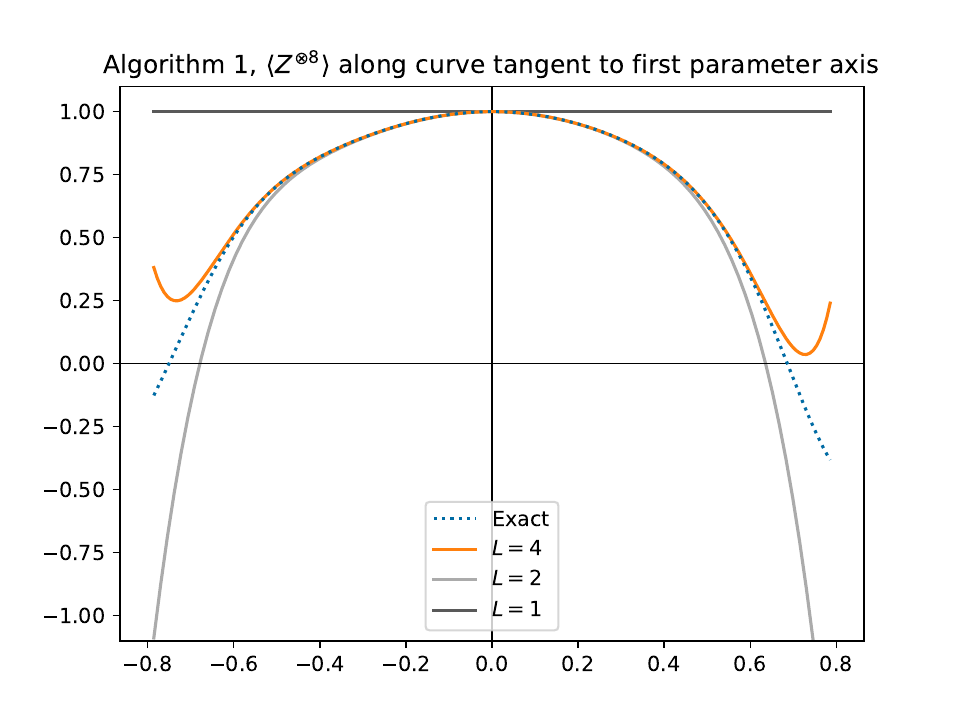}}\hspace{1em}%
	\subcaptionbox{Algorithm \ref{algorithm:fourier}, plot along curve $\gamma_2$\label{fig:deafult_fourier_tangent}}{\includegraphics[width=2.3in]{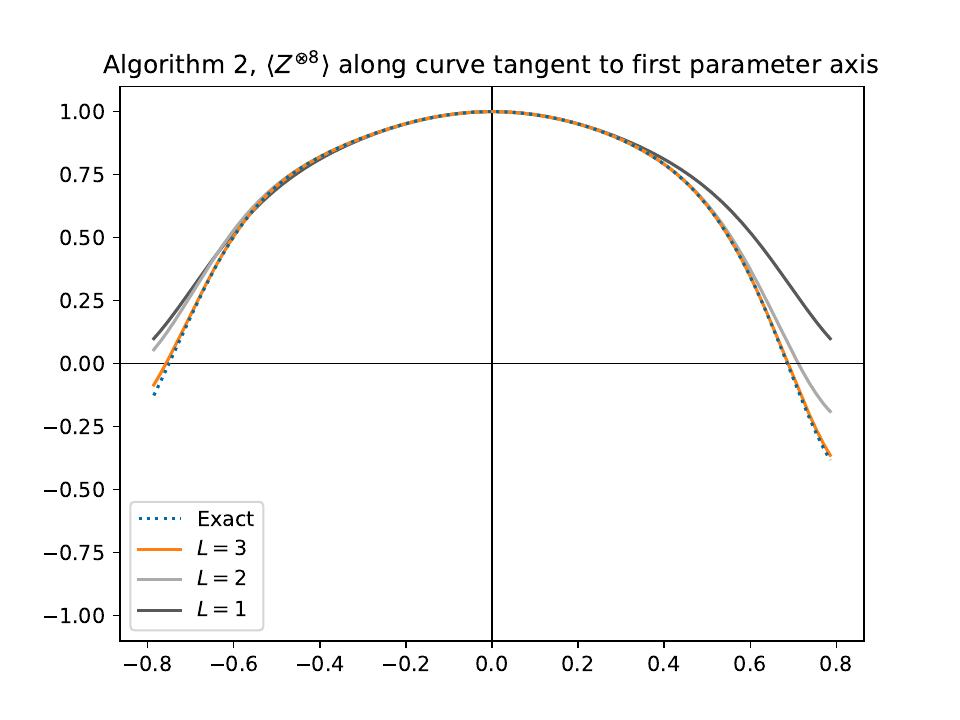}}
	
	\bigskip
	
	\centering
	\subcaptionbox{Algorithm \ref{algorithm:taylor}, plot along curve $\gamma_3$\label{fig:default_taylor_norm}}{\includegraphics[width=2.3in]{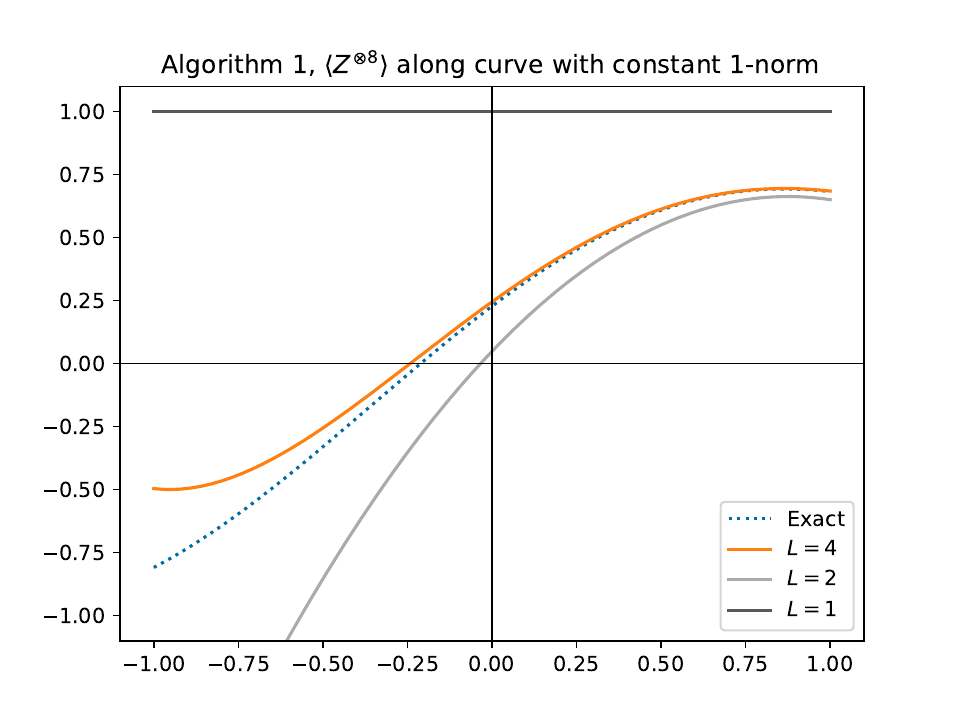}}\hspace{1em}%
	\subcaptionbox{Algorithm \ref{algorithm:fourier}, plot along curve $\gamma_3$\label{fig:default_fourier_norm}}{\includegraphics[width=2.3in]{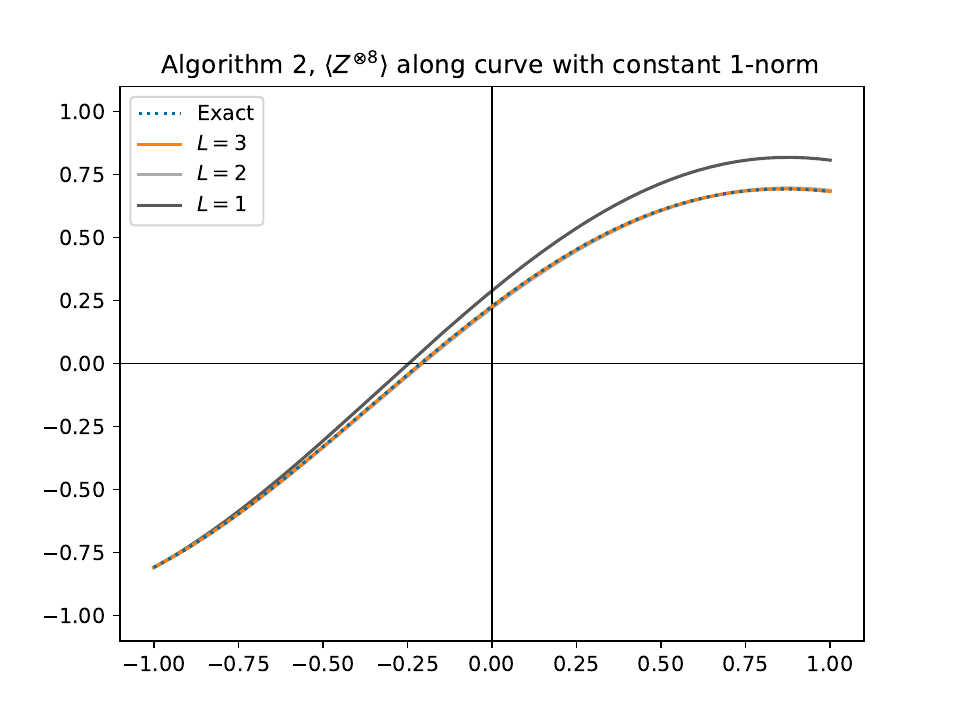}}
	\caption{These plots show the function $f$ along several curves, where $f$ is defined as in Section \ref{sec:algorithms} with respect to the circuit described in Section \ref{sec:experiments}, with $n=8$, $d=2$, and the observable $Z^{\otimes 8}$. The plots also show the approximation $\tilde{f}$ obtained by the respective algorithm for various values of $L$. The plots show the behavior close to the grid point $0\in\mathbb{R}^{16}$.}
	\label{fig:default}
\end{figure}
	
	\begin{figure}[htp]
		\centering
		\subcaptionbox{Algorithm \ref{algorithm:taylor}, plot along curve $\gamma_1$\label{fig:special_taylor_diagonal}}{\includegraphics[width=2.3in]{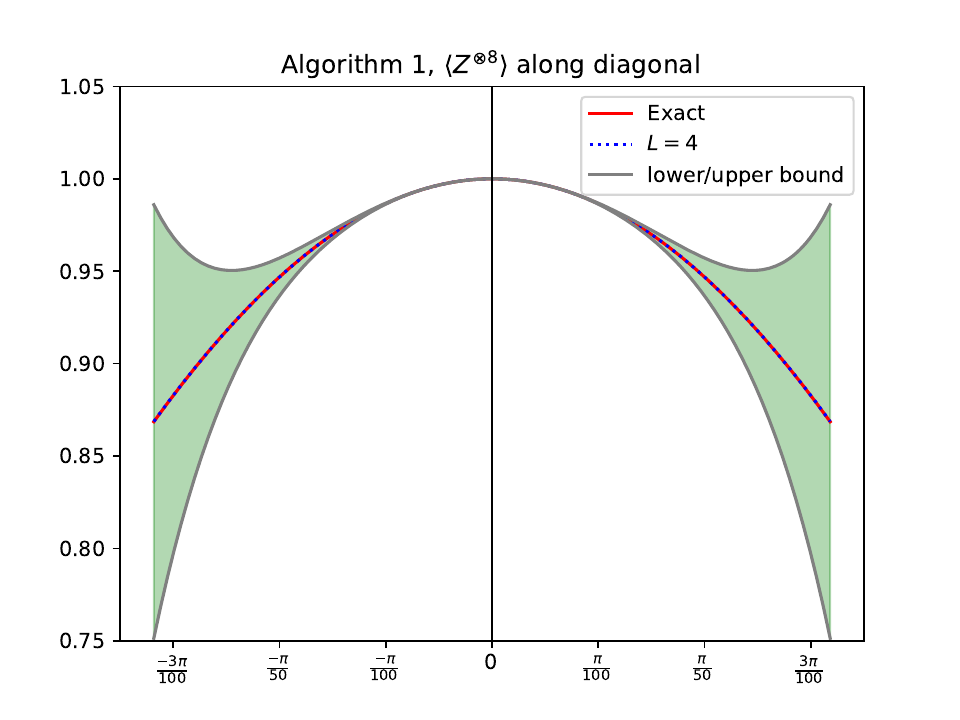}}\hspace{1em}%
		\subcaptionbox{Algorithm \ref{algorithm:fourier}, plot along curve $\gamma_1$\label{fig:special_fourier_diagonal}}{\includegraphics[width=2.3in]{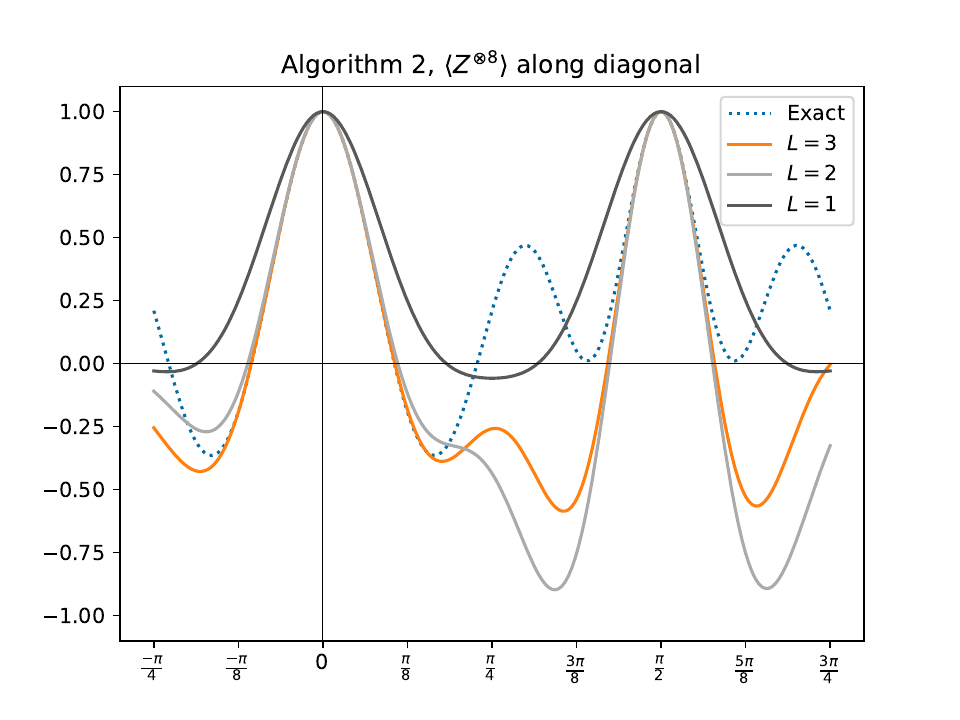}}
		
		\bigskip
		
		\centering
		\subcaptionbox{Algorithm \ref{algorithm:taylor}, plot along curve $\gamma_2$\label{fig:special_taylor_tangent}}{\includegraphics[width=2.3in]{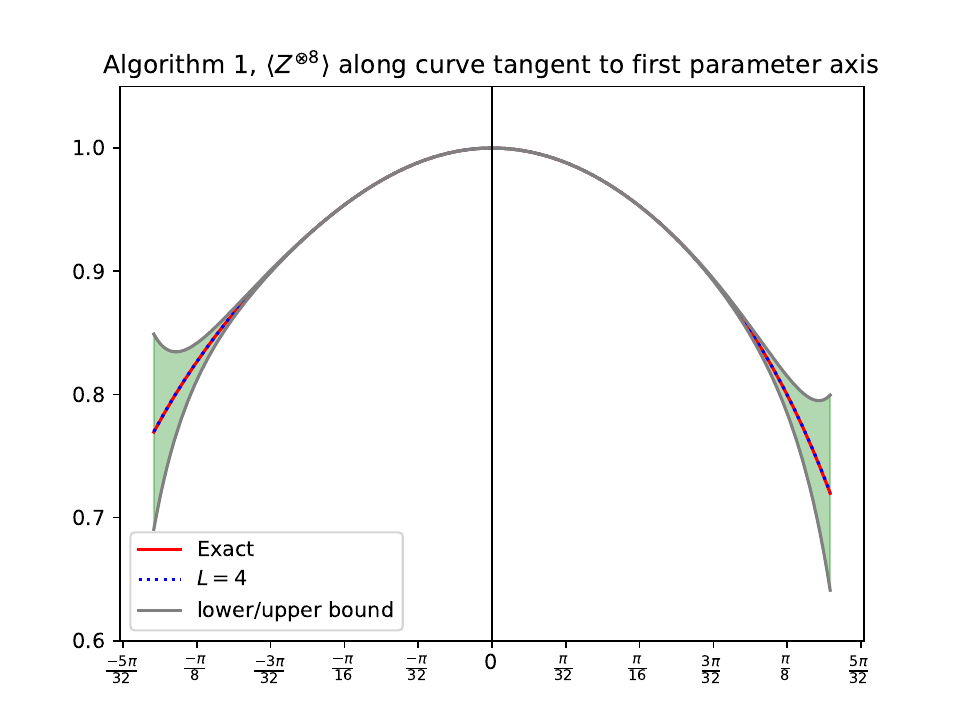}}\hspace{1em}%
		\subcaptionbox{Algorithm \ref{algorithm:fourier}, plot along curve $\gamma_2$\label{fig:special_fourier_tangent}}{\includegraphics[width=2.3in]{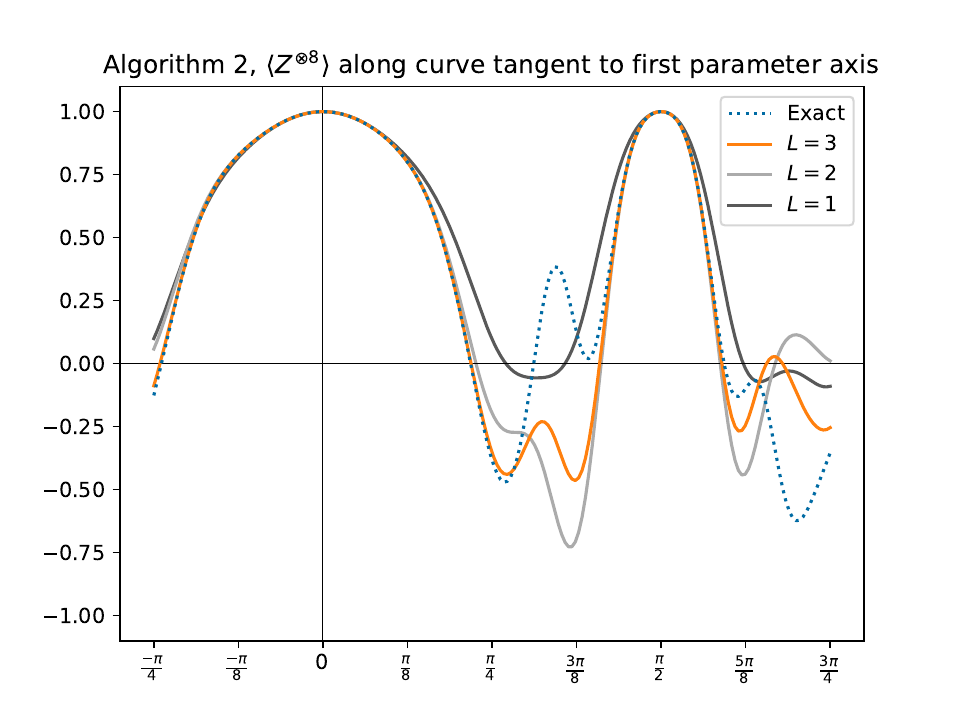}}
		
		\bigskip
		
		\centering
		\subcaptionbox{Algorithm \ref{algorithm:taylor}, plot along curve $\gamma_4$\label{fig:special_taylor_norm}}{\includegraphics[width=2.3in]{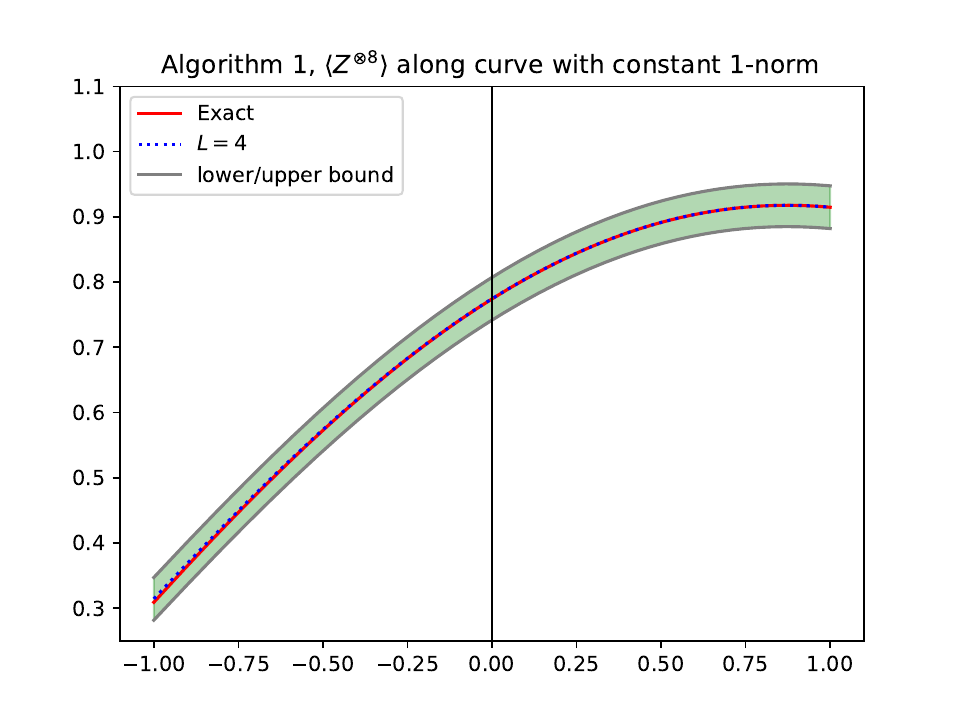}}\hspace{1em}%
		\subcaptionbox{Algorithm \ref{algorithm:fourier}, plot along curve $\gamma_5$\label{fig:special_fourier_subspace}}{\includegraphics[width=2.3in]{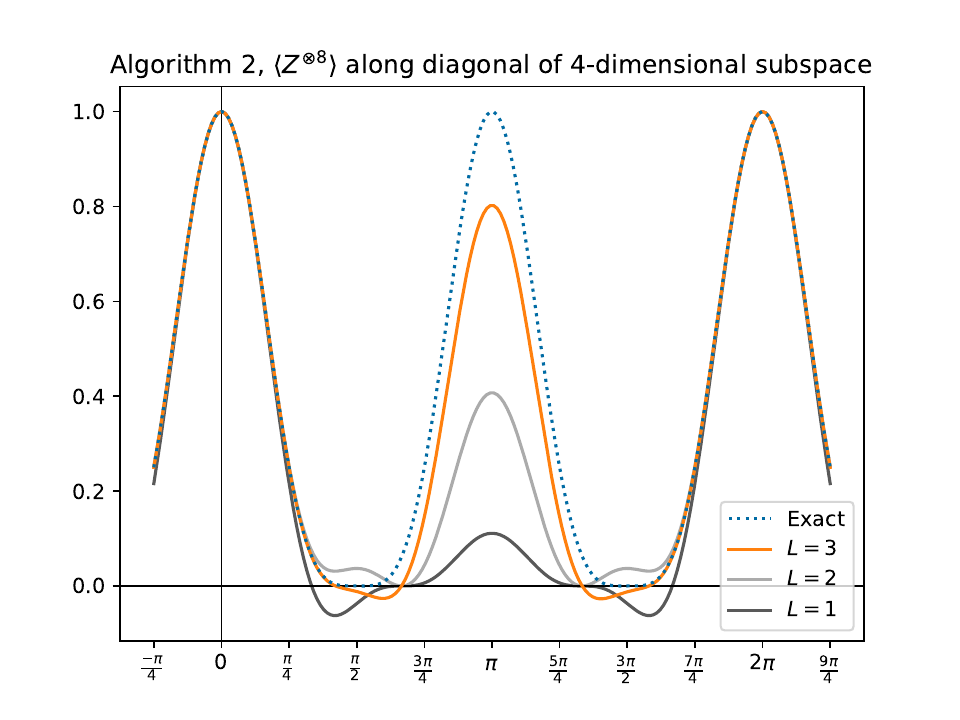}}
		\caption{These plots show the function $f$ along several curves, where $f$ is defined as in Section \ref{sec:algorithms} with respect to the circuit described in Section \ref{sec:experiments}, with $n=8$, $d=2$, and the observable $Z^{\otimes 8}$. The plots also show the approximation $\tilde{f}$ obtained by the respective algorithm for various values of $L$. For Algorithm \ref{algorithm:taylor}} the plots show the behavior close to the grid point $0\in\mathbb{R}^{16}$ as well as the error estimates guaranteed by Theorem \ref{thm:performance_guarantees_taylor}. For Algorithm \ref{algorithm:fourier} the global behavior of the approximation is shown.
		\label{fig:special}
	\end{figure}
	
	\subsection{Approximation error in $L^2$-norm}\label{subsec:montecarlonorm}
	Here, we analyse the {\emph{relative}} approximation error (measured in $L^2$-norm) of the surrogates computed by Algorithms \ref{algorithm:taylor} and \ref{algorithm:fourier} for orders $L=1,2,4$ and $L=1,2,3$, respectively: We determine the quantities
	\begin{align*}
		\frac{\Vert f-\tilde{f}\Vert_{L^2([-\pi /k,\pi /k]^{16})} }{ \Vert f\Vert_{L^2([-\pi /k,\pi /k]^{16})}}
	\end{align*}
	for $k=1,2,4,8$, where $f$ is defined as in Section \ref{sec:algorithms} with respect to the circuit described in the beginning of Section \ref{sec:experiments} and $\tilde{f}$ is an approximation computed by one of the algorithms. All occurring integrals were approximated using Monte Carlo integration with 300000 points sampled uniformly at random (for ap\-pro\-xi\-ma\-ting $\Vert f\Vert_{L^2([-\pi /k,\pi /k]^{16})}^2$) resp.\ 100000 points sampled uniformly at random (for approximating $\Vert f-\tilde{f}\Vert_{L^2([-\pi /k,\pi /k]^{16})}^2$). In the respective computations of the sample means, the (estimated) standard error of the mean never exceeded $2.1\%$ of the sample mean.
	
	We get that $\Vert f\Vert_{L^2([-\pi /k,\pi /k]^{16})}$ is approximately $2.87 \cdot 10^5$, $1.12\cdot 10^3$, $9.65\cdot 10^0$, $9.69\cdot 10^{-2}$, for $k=1,2,4,8$, respectively. The relative approximation error is shown in Table \ref{table:error}.
	
	For $k=1$, the domain of integration $[-\pi ,\pi]^{16}$ captures the global behaviour of $f$ and $\tilde{f}$ (due to periodicity). Since the order $L$ is always significantly smaller than $16$ (the dimension of the parameter space), one cannot expect good global approximation behaviour. As the value of $k$ increases and the domain of integration shrinks, one expects better and better approximation behaviour in light of the performance guarantees from Theorems \ref{thm:performance_guarantees_taylor} and \ref{thm:performance_guarantees_fourier} (which guarantee good behaviour locally around $0\in\mathbb{R}^{16}$). 
 Table \ref{table:error} clearly illustrates this expected behavior.

	\begin{table}[h]
		\centering
		\begin{tabular}{c|cccccc}
			\toprule
			\multicolumn{1}{c}{} & \multicolumn{3}{c}{\textbf{Algorithm \ref{algorithm:taylor}}} & \multicolumn{3}{c}{\textbf{Algorithm \ref{algorithm:fourier}}} \\
			\cmidrule(rl){2-4} \cmidrule(rl){5-7}
			\textbf{Domain}\\ $[-\pi /k,\pi /k]^{16}$ & {$L=1$} & {$L=2$} & {$L=4$} & {$L=1$} & {$L=2$} & {$L=3$} \\
			\midrule
			$k=1$ &\scriptsize $8.5\cdot 10^0$ &\scriptsize $2.3\cdot 10^2$ &\scriptsize $3.0\cdot 10^3$ &\scriptsize $1.0\cdot 10^0$ &\scriptsize $1.0\cdot 10^0$  &\scriptsize $1.0\cdot 10^0$\\
			$k=2$ &\scriptsize $8.5\cdot 10^0$ &\scriptsize $5.1\cdot 10^1$ &\scriptsize $1.5\cdot 10^2$ &\scriptsize $1.0\cdot 10^0$ &\scriptsize $9.8\cdot 10^{-1}$ &\scriptsize $1.0\cdot 10^0$\\
			$k=4$ &\scriptsize $3.2\cdot 10^0$ &\scriptsize $3.5\cdot 10^0$ &\scriptsize $2.2\cdot 10^0$ &\scriptsize $6.6\cdot 10^{-1}$ &\scriptsize $1.6\cdot 10^{-1}$ &\scriptsize $1.1\cdot 10^{-1}$\\
			$k=8$ &\scriptsize $5.3\cdot 10^{-1}$ &\scriptsize $1.2\cdot 10^{-1}$ &\scriptsize $1.8\cdot 10^{-2}$ &\scriptsize $1.2\cdot 10^{-1}$ &\scriptsize $8.2\cdot 10^{-3}$ &\scriptsize $1.6\cdot 10^{-3}$\\
			\bottomrule
		\end{tabular}
		\caption{The entries of the table show the relative approximation error ${\Vert f-\tilde{f}\Vert_{L^2([-\pi /k,\pi /k]^{16})} }/{ \Vert f\Vert_{L^2([-\pi /k,\pi /k]^{16})}}$, where $f$ is defined as in Section \ref{sec:algorithms} with respect to the circuit described in the beginning of Section \ref{sec:experiments}. Here, the rows $k=1,2,4,8$ of the table correspond to the domain of integration $[-\pi /k,\pi /k]^{16}$. The column headers describe which combination of algorithm and order $L$ was used to obtain the approximation $\tilde{f}$. All occurring integrals were approximated using Monte Carlo integration; the details are laid out in Section \ref{subsec:montecarlonorm}.}
		\label{table:error}
	\end{table}

	\section{Conclusion}\label{sec:conclusion}
	
	In this article we described two algorithms for obtaining classical surrogates of parametrized quantum circuits, i.e., for approximating the expected value of some observable with respect to some state computed by a given quantum circuit in dependence of the parameters of the circuit. Our algorithms are not granted whitebox access to the circuit, but instead exclusively make use of blackbox evaluations, which may either be simulated or implemented on actual quantum hardware.

 We proved performance guarantees and described two practical scenarios in which these guarantees are relevant. Finally, we conducted experiments on some representative, but not large-scale, problem instances, highlighting the approximation qualities of the algorithms along various trajectories through the parameter space. 

In the appendix below we provide a more thorough mathematical treatment of the topic, for the sake of providing a rigorous theoretical foundation of our results.

	\bibliographystyle{alphaurl}
	\bibliography{literature/bibliography}

	\begin{appendices}
	Here we provide the theoretical background for the algorithms discussed in Section \ref{sec:algorithms} and prove Theorems \ref{thm:performance_guarantees_taylor} and \ref{thm:performance_guarantees_fourier}.
	
	\section{Theoretical background}
	\label{appendix:theoretical_background}
	
	We fix a positive integer $m$ and consider the $\mathbb{R}$-vector space of real-valued functions
\begin{align*}
		H :=\Bigg\{
		& g(z)=\sum_{\omega\in\{-1,0,1\}^m} c_\omega e^{i\omega^T z}
		\,\Bigg\vert\, \overline{c_{-\omega}} = c_\omega\in \mathbb{C}\text{ for all }\omega \Bigg\}\subset\text{Maps}(\mathbb{R}^m,\mathbb{R}).
	\end{align*}	
	By restricting each element of $H$ to $[-\pi ,\pi ]^m$, we see that $H$ is isomorphic to a finite-dimensional $\mathbb{R}$-vector subspace of $L^2 ([-\pi ,\pi ]^m ;\mathbb{R})$. Hence $(H,\langle\cdot ,\cdot\rangle_H )$ is an $\mathbb{R}$-Hilbert space with induced inner product
	\begin{align*}
		\langle g_1 ,g_2\rangle_H= \int_{[-\pi ,\pi]^m} g_1 (z) g_2 (z) \mathrm{d}z, \ g_1,g_2\in H.
	\end{align*}
	Obviously, every $g\in H$ is a real-analytic function on $\mathbb{R}^m$. Moreover, the power series expansion of $g$ around any point in $\mathbb{R}^m$ converges to $g$ on all of $\mathbb{R}^m$.
	
	We now define $K\colon \mathbb{R}^m\times \mathbb{R}^m \to \mathbb{R}$ as
	\begin{align*}
		K(x,z):=\frac{1}{(2\pi )^m} \sum_{\omega\in\{-1,0,1\}^m} e^{i\omega^T(x-z)}, \ x,z\in \mathbb{R}^m,
	\end{align*}
	and for fixed $x\in\mathbb{R}^m$ we write $K_x := K(x,\cdot )$. In particular, we have $K_x\in H$ for all $x\in\mathbb{R}^m$. Note that $K$ is closely related to the multivariate Dirichlet kernel. The relevance of $K$ for our purposes lies in the following lemma:
	
	\begin{lemma}
		$(H,\langle\cdot ,\cdot\rangle_H )$ is a reproducing kernel Hilbert space and $K$ is its uniquely determined reproducing kernel. Furthermore, we have
		\begin{align*}
			K(x,z) = \frac{1}{(2\pi )^m} \prod_{j=1}^m (1+2\cos(x_j - z_j ))
		\end{align*}
		for all $x,z\in \mathbb{R}^m$.
	\end{lemma}
	\begin{proof}
		Using that, for all $\omega ,\mu\in\{-1,0,1\}^m$, we have
		\begin{align*}
			\int_{[-\pi ,\pi]^m}  e^{i(\omega -\mu)^T z} \mathrm{d}z =
			\begin{cases}
				(2\pi )^m & \text{if }\omega =\mu ,\\
				0 & \text{if }\omega \neq\mu ,\\
			\end{cases}
		\end{align*}
		a direct calculation shows that $g(x)=\langle g,K_x\rangle_H$ for all $g\in H,x\in \mathbb{R}^m$. In particular, by virtue of the Cauchy-Schwarz inequality, the evaluation functional at $x$ is a bounded linear operator on $H$ for all $x\in\mathbb{R}^m$, i.e., $H$ is a reproducing kernel Hilbert space. Uniqueness in the Riesz representation theorem then implies that $K$ is the uniquely determined reproducing kernel for $H$. It remains to prove the product formula for $K$, which again follows from an easy calculation. Indeed, we have
		\begin{align*}
			K(x,z)
			& =  \frac{1}{(2\pi )^m} \sum_{\omega\in\{-1,0,1\}^m} e^{i\omega^T(x-z)}\\
			& = \frac{1}{(2\pi )^m} \prod_{j=1}^m\left(\sum_{\omega_j\in\{-1,0,1\}} e^{i\omega_j  (x_j -z_j )}\right)\\
			& = \frac{1}{(2\pi )^m} \prod_{j=1}^m\left(e^{-i (x_j -z_j )} + 1 + e^{i(x_j -z_j )}\right)\\
			& = \frac{1}{(2\pi )^m} \prod_{j=1}^m (1+2\cos(x_j - z_j ))
		\end{align*}
        for $x,z\in\mathbb{R}^m$.
	\end{proof}

	\begin{notation}
		We introduce some notation:
		\begin{itemize}
			\label{notation:shiftvectors}
			\item Let $\mathbb{Z}/4\mathbb{Z}$ denote the ring of integers modulo $4$ and, for $x\in\mathbb{Z}$, denote its residue class as $\overline{x}$. We now fix the complete residue system $S:=\{0,1,2,3\}$ modulo $4$, i.e., for any $x\in\mathbb{Z}$ there is a unique $y\in S$ satisfying $\overline{x}=\overline{y}$. In the following, we denote $y$ as $x\bmod 4$.
			\item Let $\alpha\in (\mathbb{Z}_{\geq 0})^m$ be a multiindex with $|\alpha |:=\alpha_1 + \dots + \alpha_m=k$ and let $\mathfrak{i}\in\{-1,1\}^{k}$ be a tuple of length $|\alpha |$ with entries in $\{-1,1\}$. Using multiindex notation, the product of the elements of $\mathfrak{i}$ can be written as $\mathfrak{i}^{(1,\dots ,1)}$. Furthermore, given that $k=\alpha_1 + \dots + \alpha_m$, we can index the entries of $\mathfrak{i}$ as
			\begin{align*}
				\mathfrak{i} = \left(\mathfrak{i}_{1,1},\dots ,\mathfrak{i}_{1,\alpha_1},\dots ,\mathfrak{i}_{m,1},\dots ,\mathfrak{i}_{m,\alpha_m}\right).
			\end{align*}
			We then define $p_{\alpha ,\mathfrak{i}}\in \frac{\pi}{2}S^m\subseteq\mathbb{R}^m$ as
			\begin{align*}
				p_{\alpha ,\mathfrak{i}} = \frac{\pi}{2}\cdot\left((\mathfrak{i}_{1,1}+\dots +\mathfrak{i}_{1,\alpha_1})\bmod 4,\dots ,(\mathfrak{i}_{m,1}+\dots +\mathfrak{i}_{m,\alpha_m})\bmod 4\right).
			\end{align*}
			Note that it is possible that $\alpha_j =0$ for some $j\in\{1,\dots ,m\}$, in which case $\left(\mathfrak{i}_{j,1},\dots ,\mathfrak{i}_{j,\alpha_j}\right)$ is the empty tuple and $\mathfrak{i}_{j,1}+\dots +\mathfrak{i}_{j,\alpha_j}$ is the empty sum, which is $0$ by convention.
		\end{itemize}
	\end{notation}

	\begin{lemma}
		\label{lemma:parametershift}
		Let $g\in H$. For all $\alpha\in (\mathbb{Z}_{\geq 0})^m$, $x\in \mathbb{R}^m$ (see Notation \ref{notation:shiftvectors}) we then have
		\begin{align*}
			D^\alpha g (x) & = \frac{1}{2^{|\alpha |}} \sum_{\mathfrak{i}\in\{-1,1\}^{|\alpha |}} \mathfrak{i}^{(1,\dots ,1)} g(x+p_{\alpha ,\mathfrak{i}}) .
		\end{align*}
		In particular, with $\alpha!:=\prod\alpha_i!$, we have
		\begin{align*}
			g(x) = \sum_{k=0}^{\infty}  \sum_{\alpha\in (\mathbb{Z}_{\geq 0})^m\colon |\alpha |=k}\frac{x^\alpha}{\alpha! \cdot 2^{|\alpha |}} \left( \sum_{\mathfrak{i}\in\{-1,1\}^{|\alpha |}} \mathfrak{i}^{(1,\dots ,1)} g(p_{\alpha ,\mathfrak{i}}) \right)
		\end{align*}
		for all $x\in\mathbb{R}^m$.
	\end{lemma}
	\begin{remark}
		Note that the assumptions in Lemma \ref{lemma:parametershift} do not exclude the case that $\alpha =(0,\dots ,0)$. With the usual conventions for empty sums and empty products and noting that, set-theoretically, $\{-1,1\}^0 = \{\emptyset\}$, the claimed expression for $D^{(0,\dots ,0)} g(x)$ indeed reduces to $g(x)$.
	\end{remark}
	\begin{proof}
		The second equality follows from the first by expanding $g$ as a power series around $0\in\mathbb{R}^m$ and plugging in, so it suffices to prove the first equality. But the first equality follows easily from direct computation using induction on $|\alpha|$. We omit the details.
	\end{proof}

\begin{remark}
	  In the context of quantum machine learning, the first equality in Lemma \ref{lemma:parametershift} is called a \emph{parameter shift rule}, see e.g. \cite{Mitarai_2018}, \cite{Schuld_2019}, \cite{Mari_2021}, \cite{Wierichs2022generalparameter}. 
\end{remark}
	
	\begin{lemma}
		\label{lemma:taylorerror}
		Let $g\in H$, let $L\in\mathbb{Z}_{\geq 0}$ and, with the usual multiindex notation, let $T_L g\colon\mathbb{R}^m\to\mathbb{R}$,
		\begin{align*}
			(T_L g)(x) = \sum_{\alpha\in (\mathbb{Z}_{\geq 0})^m\colon |\alpha |\leq L} \frac{D^\alpha g (0)}{\alpha !}x^\alpha
			\text{ for all }x\in\mathbb{R}^m ,
		\end{align*}
		be the $L$-th order Taylor polynomial of $g$ at $0\in\mathbb{R}^m$. Then the following estimate for the approximation error holds for all $x\in\mathbb{R}^m$:
		\begin{align*}
			\left|g(x)-(T_L g)(x)\right|
			\leq \left(\exp\left(\Vert x\Vert_1\right)
			-
			\sum_{k=0}^{L} \frac{\Vert x\Vert_1^k}{k!}\right) \Vert g\Vert_\infty
			,
		\end{align*}
		where $\Vert\cdot\Vert_1$ denotes the $1$-norm on $\mathbb{R}^m$ and $\Vert\cdot\Vert_\infty$ denotes the sup norm on $H$. In particular, if $\Vert x\Vert_1\leq 1 +\frac{L}{2}$, we have
		\begin{align*}
			\left|g(x)-(T_L g)(x)\right|
			\leq
			2\frac{\Vert x\Vert_1^{L+1}}{(L+1)!} \Vert g\Vert_\infty
			.
		\end{align*}
	\end{lemma}
	\begin{proof}
		The second estimate follows from the first estimate with the well-known remainder term estimates for the exponential function, so it suffices to prove the first estimate. Let $x\in\mathbb{R}^m$. Using Lemma \ref{lemma:parametershift} and the multinomial theorem, we calculate
		\begin{align*}
			\left|g(x)-(T_L g)(x)\right|
			& = \left|\sum_{k=L+1}^{\infty}\sum_{\alpha\in (\mathbb{Z}_{\geq 0})^m\colon |\alpha |=k} \frac{D^\alpha g (0)}{\alpha !}x^\alpha\right|\\
			& = \left|\sum_{k=L+1}^{\infty}\sum_{\alpha\in (\mathbb{Z}_{\geq 0})^m\colon |\alpha |=k} \frac{x^\alpha}{\alpha! \cdot 2^{k}} \left( \sum_{\mathfrak{i}\in\{-1,1\}^{k}} \mathfrak{i}^{(1,\dots ,1)} g(p_{\alpha ,\mathfrak{i}}) \right)\right|\\
			& \leq \sum_{k=L+1}^{\infty}\sum_{\alpha\in (\mathbb{Z}_{\geq 0})^m\colon |\alpha |=k} \frac{|x^\alpha |}{\alpha! \cdot 2^{k}}  \sum_{\mathfrak{i}\in\{-1,1\}^{k}} \left|\mathfrak{i}^{(1,\dots ,1)} g(p_{\alpha ,\mathfrak{i}})\right|\\
			& \leq \Vert g\Vert_\infty\sum_{k=L+1}^{\infty}\sum_{\alpha\in (\mathbb{Z}_{\geq 0})^m\colon |\alpha |=k} \frac{|x^\alpha |}{\alpha! \cdot 2^{k}}  \sum_{\mathfrak{i}\in\{-1,1\}^{k}} 1\\
			& = \Vert g\Vert_\infty\sum_{k=L+1}^{\infty}\frac{1}{k!}\sum_{\alpha\in (\mathbb{Z}_{\geq 0})^m\colon |\alpha |=k}k! \frac{|x_1|^{\alpha_1}\cdots |x_m|^{\alpha_m}}{\alpha_1!\cdots \alpha_m!}\\
			& = \Vert g\Vert_\infty\sum_{k=L+1}^{\infty}\frac{1}{k!}(|x_1 |+\dots +|x_m|)^k .
		\end{align*}
		The claim follows.
	\end{proof}
	
	\begin{notation}
		\label{notation:orthogonalProjection}
		If $V$ is an $\mathbb{R}$-vector subspace of $H$, then we denote the orthogonal projection from $H$ onto $V$ wrt.\ $\langle\cdot ,\cdot\rangle_H$ as $\mathcal{P}_V\colon H\to V$. Note that $\mathcal{P}_V$ necessarily exists and is unique, since $H$ is finite-dimensional.
	\end{notation}
	
	\begin{lemma}
		\label{lemma:lgsHasSolution}
		Let $D\in\mathbb{Z}_{\geq 1}$ and let $p_1,\dots ,p_D\in\mathbb{R}^m$. Then, for $g\in H$, the following hold:
		\begin{enumerate}[(i)]
			\item \label{lemma:lgsHasSolution_i}
			The linear system of equations
			\begin{align*}
				\left(K(p_i , p_j)\right)_{1\leq i,j\leq D}
				\cdot\eta
				=
				\begin{pmatrix}
					g(p_1) \\ \vdots \\ g(p_D)
				\end{pmatrix}
				,
				\text{ where }\eta \in\mathbb{R}^D
			\end{align*}
			has at least one solution.
			\item \label{lemma:lgsHasSolution_ii}
			While the solution $\eta$ to the linear system of equations in (\ref{lemma:lgsHasSolution_i}) is not unique in general, $\sum_{j=1}^{D} \eta_j K_{p_j}\in H$ is uniquely determined. More precisely, the set
			\begin{align*}
				\left\{\sum_{j=1}^{D} \eta_j K_{p_j}\Bigg\vert
				\eta\in\mathbb{R}^D\text{ and }
				\left(K(p_i , p_j)\right)_{1\leq i,j\leq D}
				\cdot\eta
				=
				\begin{pmatrix}
					g(p_1) \\ \vdots \\ g(p_D)
				\end{pmatrix}
				\right\}\subseteq H
			\end{align*}
			has exactly one element.
			\item \label{lemma:lgsHasSolution_iii}
			$\eta\in\mathbb{R}^D$ solves the linear system of equations in (\ref{lemma:lgsHasSolution_i}) if and only if
			\begin{align*}
				\sum_{j=1}^{D} \eta_j K_{p_j} = \mathcal{P}_{V}(g),
			\end{align*}
			where $V=\operatorname{span}_{\mathbb{R}}(\{K_{p_1},\dots ,K_{p_D}\})$ (see Notation \ref{notation:orthogonalProjection}).
			\item \label{lemma:lgsHasSolution_iv}
			If $p\in\mathbb{R}^m$ satisfies $K_p\in V$, then we have $(\mathcal{P}_V (g))(p) = g(p)$.
		\end{enumerate}
	\end{lemma}
	\begin{proof}
		By existence and uniqueness of the orthogonal projection onto closed subspaces of Hilbert spaces, (\ref{lemma:lgsHasSolution_i}) and (\ref{lemma:lgsHasSolution_ii}) readily follow from (\ref{lemma:lgsHasSolution_iii}), so it suffices to show (\ref{lemma:lgsHasSolution_iii}) and (\ref{lemma:lgsHasSolution_iv}).
		
		We start with proving (\ref{lemma:lgsHasSolution_iii}). First assume that $\eta\in\mathbb{R}^D$ satisfies $\sum_{j=1}^{D} \eta_j K_{p_j} = \mathcal{P}_{V}(g)$. Since $g-\mathcal{P}_V (g)$ is orthogonal to $V$, we get
		\begin{align*}
			0
			& = \langle K_{p_i},  g-\mathcal{P}_V (g)\rangle_H\\
			& = \langle K_{p_i},  g\rangle_H -\left\langle K_{p_i},\sum_{j=1}^{D} \eta_j K_{p_j}\right\rangle_H\\
			& = g(p_i) - \sum_{j=1}^{D} K(p_i , p_j )\eta_j
		\end{align*}
		for all $i\in\{1,\dots ,D\}$. Hence $\eta$ solves the linear system of equations in (\ref{lemma:lgsHasSolution_i}), as desired.
		
		Turning our attention to the other implication, we now assume that $\eta\in\mathbb{R}^D$  solves the linear system of equations in (\ref{lemma:lgsHasSolution_i}). Pick $\tilde{\eta}\in\mathbb{R}^D$, such that $\sum_{j=1}^{D} \tilde{\eta}_j K_{p_j} = \mathcal{P}_{V}(g)$. But then, by the already proven first implication, we get
		\begin{align*}
			\left\langle K_{p_i},  \sum_{j=1}^{D} {\eta}_j K_{p_j}-\mathcal{P}_{V}(g)\right\rangle_H
			& = \left\langle K_{p_i},\sum_{j=1}^{D} \eta_j K_{p_j}\right\rangle_H - \left\langle K_{p_i},\sum_{j=1}^{D} \tilde{\eta}_j K_{p_j}\right\rangle_H\\
			& = \sum_{j=1}^{D} K(p_i , p_j )\eta_j - \sum_{j=1}^{D} K(p_i , p_j )\tilde{\eta}_j\\
			& = g(p_i) - g(p_i) = 0
		\end{align*}
		for all $i\in\{1,\dots ,D\}$. It follows that $\sum_{j=1}^{D} {\eta}_j K_{p_j}-\mathcal{P}_{V}(g)$ is orthogonal to $V$. But $\sum_{j=1}^{D} {\eta}_j K_{p_j}-\mathcal{P}_{V}(g)$ is itself contained in $V$. It follows that $\sum_{j=1}^{D} \eta_j K_{p_j} = \mathcal{P}_{V}(g)$, as desired.
		
		It remains to show (\ref{lemma:lgsHasSolution_iv}). Let $p\in\mathbb{R}^m$ with $K_p\in V$. Since $g-\mathcal{P}_V (g)$ is orthogonal to $V$, we get
		\begin{align*}
			0
			& = \langle K_p,  g-\mathcal{P}_V (g)\rangle_H\\
			& = g(p) - (\mathcal{P}_V (g))(p),
		\end{align*}
		and the claim follows.
	\end{proof}
	
		\begin{lemma}
		\label{lemma:fewer_tuples_enough_for_fourier}
		Let $d\in\{0,1,\dots ,m\}$ and let $\mathcal{I}_d$ be the set of all tuples in $\mathbb{R}^m$, for which at most $d$ entries are non-zero, i.e.,
		\begin{align*}
			\mathcal{I}_d = \left\{z\in\mathbb{R}^m\vert \operatorname{card}(\{j\in\{1,\dots ,m\}\vert z_j\neq 0\})\leq d\right\} .
		\end{align*}
		Then, for all $p\in (\frac{\pi}{2}S^m )\cap \mathcal{I}_d$, we have
		\begin{align*}
			K_p\in\operatorname{span}_{\mathbb{R}}\left(\left\{K_q\Big\vert q\in \left(\frac{\pi}{2}\{-1, 0, 1\}^m\right)\cap\mathcal{I}_d\right\}\right) .
		\end{align*}
	\end{lemma}
	\begin{proof}
		Let $p\in (\frac{\pi}{2}S^m )\cap \mathcal{I}_d$. For $j\in\{1,\dots , m\}$, let
		\begin{align*}
			\tilde{p}_j =
			\begin{cases}
				-\frac{\pi}{2} & \text{if }p_j = 3\frac{\pi}{2},\\
				p_j & \text{otherwise}.
			\end{cases}
		\end{align*}
		
		and let $\tilde{p}=(\tilde{p}_1,\dots ,\tilde{p}_m)\in (\frac{\pi}{2}\{-1,0,1,2\}^m )\cap\mathcal{I}_d$. By direct computation exploiting periodicity one readily verifies that $K_p = K_{\tilde{p}}$, so we can consider $K_{\tilde{p}}$ instead of $K_p$. If $\tilde{p}_j\neq\pi$ for all $j\in\{1,\dots ,m\}$, then we are done, so assume we find some $i\in\{1,\dots ,m\}$ such that $\tilde{p}_i =\pi$. Again by direct computation we get
		\begin{align*}
			K_{\tilde{p}} = 
			& K_{(\tilde{p}_1,\dots ,\tilde{p}_{i-1} , \frac{\pi}{2}, \tilde{p}_{i+1},\dots ,\tilde{p}_m)}+K_{(\tilde{p}_1,\dots ,\tilde{p}_{i-1} , -\frac{\pi}{2}, \tilde{p}_{i+1},\dots ,\tilde{p}_m)}\\
			& -K_{(\tilde{p}_1,\dots ,\tilde{p}_{i-1} , 0, \tilde{p}_{i+1},\dots ,\tilde{p}_m)}
		\end{align*}
		The claim now follows by an inductive argument.
	\end{proof}
	\begin{remark}
		In \cite{Mari_2021}, a similar relationship between values at shifted arguments was described in the section dealing with arbitrary-order derivatives.
	\end{remark}
	
	\begin{lemma}
		\label{lemma:basis_for_H}
		The family $\left(K_q\right)_{q\in \frac{\pi}{2}\{-1, 0, 1\}^m}$ is a basis for $H$.
	\end{lemma}
	\begin{proof}
		Since $\dim_{\mathbb{R}}(H)=3^m$, it suffices to show that $\left\{K_q\vert q\in \frac{\pi}{2}\{-1, 0, 1\}^m\right\}$ spans $H$. Invoking Lemma \ref{lemma:fewer_tuples_enough_for_fourier} in the special case $d=m$, we further reduce this to showing that $\{K_p\vert p\in\frac{\pi}{2}S^m\}$ spans $H$. To this end, let $g\in H$. By direct computation one readily verifies that
		\begin{align*}
			g = \left(\frac{\pi}{2}\right)^m\sum_{p\in\frac{\pi}{2}S^m} \langle g, K_p\rangle_H K_p
		\end{align*}
		The claim follows.
	\end{proof}

	\begin{lemma}
		\label{lemma:basis_for_subspaces_of_H}
		Let $d\in\{0,1,\dots ,m\}$ and let $\mathcal{I}_d$ be as in Lemma \ref{lemma:fewer_tuples_enough_for_fourier}. Then the family $\left(K_q\right)_{q\in \left(\frac{\pi}{2}\{-1, 0, 1\}^m\right)\cap\mathcal{I}_d}$ is a basis for $\operatorname{span}_{\mathbb{R}}(\{K_p\vert p\in\mathcal{I}_d\})$.
	\end{lemma}
	\begin{proof}
		Linear independence follows from Lemma \ref{lemma:basis_for_H}. Let $p\in\mathcal{I}_d$. It suffices to prove that $K_p$ is contained in $\operatorname{span}_{\mathbb{R}}(\{K_q\vert q\in \left(\frac{\pi}{2}\{-1, 0, 1\}^m\right)\cap\mathcal{I}_d\})$. Invoking Lemma \ref{lemma:fewer_tuples_enough_for_fourier}, we further reduce this to showing that $K_p$ is contained in $\operatorname{span}_{\mathbb{R}}(\{K_q\vert q\in \left(\frac{\pi}{2} S^m\right)\cap\mathcal{I}_d\})$. Since $p\in\mathcal{I}_d$, we find a (not necessarily uniquely determined) set $I_{p,d}\subseteq\{1,\dots ,m\}$ with cardinality $\operatorname{card}(I_{p,d})=d$, such that $\{j\in\{1,\dots ,m\}\vert p_j\neq 0\}\subseteq I_{p,d}$. We then define 
		\begin{align*}
			\mathcal{S}_{p,d} =\{z\in\mathbb{R}^m\vert z_j =0\text{ for all } j\in\{1,\dots ,m\}\setminus I_{p,d}\} .
		\end{align*}
		Since $\mathcal{S}_{p,d}\subseteq \mathcal{I}_d$, it suffices to prove that $K_p$ is contained in $\operatorname{span}_{\mathbb{R}}(\{K_q\vert q\in \left(\frac{\pi}{2} S^m\right)\cap\mathcal{S}_{p,d}\})$. By direct computation one readily verifies that
		\begin{align*}
			K_p
			=
			\frac{(2\pi )^m}{4^d\cdot 3^{m-d}}
			\sum_{q\in \left(\frac{\pi}{2} S^m\right)\cap\mathcal{S}_{p,d}} \langle K_p ,K_q\rangle_H K_q
			.
		\end{align*}
		The claim follows.
	\end{proof}
	
	\section{Proof of Theorem \ref{thm:performance_guarantees_taylor}}
	\label{appendix:proof_taylor}
	
	We now provide a proof for Theorem \ref{thm:performance_guarantees_taylor}. To this end, adopt the notation from Section \ref{sec:algorithms} and Theorem \ref{thm:performance_guarantees_taylor}.
	
	We first prove (\ref{thm:performance_guarantees_taylor_number_f_eval}), so assume $L\leq m$. In accordance with Lemma \ref{lemma:partials_for_our_function_f}, the number of evaluations of $f$ is clearly bounded from above by
	\begin{align*}
		\sum_{\alpha\in (\mathbb{Z}_{\geq 0})^m\text{ with }|\alpha |\leq L} 2^{|\alpha |}
		& \leq \sum_{\alpha\in (\mathbb{Z}_{\geq 0})^m\text{ with }|\alpha |\leq L} 2^L\\
		& = 2^L\cdot \operatorname{card}\left(\left\{\alpha\in (\mathbb{Z}_{\geq 0})^m\colon |\alpha |\leq L\right\}\right)\\
		& = 2^L\cdot \binom{m+L}{L}\\
		& = 2^L\cdot\frac{(m+1)\cdots (m+L)}{L!}\\
		& \leq 2^L\cdot\frac{(2m)^L}{L!}\\
		& = \frac{4^L}{L!}m^L
		,
	\end{align*}
	as desired.

	We now turn our attention to (\ref{thm:performance_guarantees_taylor_error_estimate}). We calculate for all $\theta\in\mathbb{R}^m$:
	\begin{align*}
		|f(\theta)|
		& = \left| \langle0^n|U^\dag (\theta)\mathcal{M}U(\theta)|0^n\rangle \right| \\
		& = \left|\sum_{(P_1 ,\dots ,P_n)\in\{I,X,Y,Z\}^n} a_{(P_1 ,\dots ,P_n)}\langle0^n|U^\dag (\theta)(P_1\otimes\cdots\otimes P_n )U(\theta)|0^n\rangle \right|\\
		& \leq \sum_{(P_1 ,\dots ,P_n)\in\{I,X,Y,Z\}^n} \left| a_{(P_1 ,\dots ,P_n)}\right| \left| \langle0^n|U^\dag (\theta)(P_1\otimes\cdots\otimes P_n )U(\theta)|0^n\rangle \right|  
		.
	\end{align*}
	Noting that $\left| \langle0^n|U^\dag (\theta)U(\theta)|0^n\rangle \right|=1$ and that the eigenvalues of the Hermitian matrix $P_1\otimes\cdots\otimes P_n$ have absolute value $1$ for all $(P_1 ,\dots ,P_n)\in\{I,X,Y,Z\}^n$, we then get
	\begin{align*}
		\Vert f\Vert_\infty
		\leq \sum_{(P_1 ,\dots ,P_n)\in\{I,X,Y,Z\}^n} \left| a_{(P_1 ,\dots ,P_n)}\right|
		.
	\end{align*}
	Since we clearly have $f\in H$, the desired estimates now follow from Lemma \ref{lemma:taylorerror}. This concludes the proof of Theorem \ref{thm:performance_guarantees_taylor}.
	
	\section{Proof of Theorem \ref{thm:performance_guarantees_fourier}}
	\label{appendix:proof_fourier}
	We now provide a proof for Theorem \ref{thm:performance_guarantees_fourier}. To this end, adopt the notation from Section \ref{sec:algorithms} and Theorem \ref{thm:performance_guarantees_fourier}. In particular, $\{p_1 ,\dots , p_D\} = \left(\frac{\pi}{2}\{-1, 0, 1\}^m\right)\cap\mathcal{I}_L$ and $D$ is the cardinality of this set.
	
	We first prove (\ref{thm:performance_guarantees_fourier_number_f_eval}). For any tuple $z$ in the set $\left(\frac{\pi}{2}\{-1, 0, 1\}^m\right)\cap\mathcal{I}_L$, there exists a (not necessarily uniquely determined) subset $I_z$ of $\{1,\dots ,m\}$ with cardinality $\operatorname{card}(I_z) = L$, such that $\{j\in\{1,\dots ,m\}\vert z_j\neq 0\}\subseteq I_z$. Any entry $z_j$ for $j\in I_z$ is contained in $\{-\frac{\pi}{2}, 0, \frac{\pi}{2}\}$, and all other entries are $0$. So, since there are precisely $\binom{m}{L}$ subsets of $\{1,\dots ,m\}$ with cardinality $L$, the number of evaluations of $f$ is clearly bounded from above by
	\begin{align*}
		3^L\cdot\binom{m}{L}\leq \frac{3^L}{L!}m^L
		.
	\end{align*}

	Since $f\in H$, (\ref{thm:performance_guarantees_fourier_projection}) is obvious from (\ref{lemma:lgsHasSolution_iii}) in Lemma \ref{lemma:lgsHasSolution}, so we move on to proving (\ref{thm:performance_guarantees_fourier_minimal_norm}). To this end, let
	\begin{align*}
		V:=\operatorname{span}_{\mathbb{R}}(\{K_{p_1},\dots ,K_{p_D}\})
		=
		\operatorname{span}_{\mathbb{R}}\left(\left\{K_q\Big\vert q\in \left(\frac{\pi}{2}\{-1, 0, 1\}^m\right)\cap\mathcal{I}_L\right\}\right).
	\end{align*}
	In particular, we have $\tilde{f}=\mathcal{P}_V (f)$ and $\tilde{f}$ agrees with $f$ on $\left(\frac{\pi}{2}\{-1, 0, 1\}^m\right)\cap\mathcal{I}_L$ by (\ref{lemma:lgsHasSolution_iv}) in Lemma \ref{lemma:lgsHasSolution}. Now assume that $h\in H$ agrees with $f$ on $\left(\frac{\pi}{2}\{-1, 0, 1\}^m\right)\cap\mathcal{I}_L$. But then, by invoking Lemma \ref{lemma:lgsHasSolution} yet again, we get that $\mathcal{P}_V (h)=\tilde{f}$. This implies that $\Vert \tilde{f}\Vert_H \leq \Vert h\Vert_H$ with equality if and only if $h=\tilde{f}$. The claim follows.
	
	We now turn our attention towards (\ref{thm:performance_guarantees_fourier_coincides_on_I_L}). Let $p\in\mathcal{I}_L$. Applying Lemma \ref{lemma:basis_for_subspaces_of_H} with $d=L$ we get
	\begin{align*}
		K_p\in \operatorname{span}_{\mathbb{R}}\left(\left\{K_q\Big\vert q\in \left(\frac{\pi}{2}\{-1, 0, 1\}^m\right)\cap\mathcal{I}_L\right\}\right)=V.
	\end{align*}
	By (\ref{lemma:lgsHasSolution_iv}) in Lemma \ref{lemma:lgsHasSolution} we then get $\tilde{f}(p)=(\mathcal{P}_V (f))(p)=f(p)$, as desired. The claim follows.
	
	(\ref{thm:performance_guarantees_fourier_coordinate_axes}) and (\ref{thm:performance_guarantees_fourier_whole_space}) readily follow from (\ref{thm:performance_guarantees_fourier_coincides_on_I_L}), so it remains to prove (\ref{thm:performance_guarantees_fourier_error_estimate}). To this end, let $\alpha\in (\mathbb{Z}_{\geq 0})^m$ with $|\alpha |\leq L$. By Lemma \ref{lemma:parametershift}, we have
	\begin{align*}
		D^\alpha (f-\tilde{f}) (0) & = \frac{1}{2^{|\alpha |}} \sum_{\mathfrak{i}\in\{-1,1\}^{|\alpha |}} \mathfrak{i}^{(1,\dots ,1)} (f(p_{\alpha ,\mathfrak{i}}) - \tilde{f}(p_{\alpha ,\mathfrak{i}})),
	\end{align*}
	since $\tilde{f}\in H$ (see Notation \ref{notation:shiftvectors}). Let $\mathfrak{i}\in\{-1,1\}^{|\alpha |}$. It suffices to show that $f(p_{\alpha ,\mathfrak{i}}) = \tilde{f}(p_{\alpha ,\mathfrak{i}})$. But, since $|\alpha |\leq L$, we have $p_{\alpha ,\mathfrak{i}}\in\mathcal{I}_L$. The claim now follows, since $f$ and $\tilde{f}$ coincide on $\mathcal{I}_L$.
	
	\end{appendices}

\end{document}